\documentclass{article}
\usepackage{graphicx} 

\usepackage{style}
\usepackage{authblk}

\newcommand{\bmka}{0}

\title{Decomposing Gaussians with Unknown Covariance}
\author[1]{Ameer Dharamshi}
\author[2]{Anna Neufeld}
\author[3]{Lucy L. Gao}
\author[4]{Jacob Bien}
\author[1,5]{Daniela Witten}
\affil[1]{Department of Biostatistics, University of Washington}
\affil[2]{Department of Mathematics and Statistics, Williams College}
\affil[3]{Department of Statistics, University of British Columbia}
\affil[4]{Department of Data Sciences and Operations, University of Southern California}
\affil[5]{Department of Statistics, University of Washington}

\begin{document}

\maketitle

\begin{abstract} 
Common workflows in machine learning and statistics rely on the ability to partition the information in a data set into independent portions. 
Recent work has shown that this may be possible even when conventional sample splitting is not  (e.g., when the number of samples $n=1$, or when observations are not independent and identically distributed).  However, the  approaches that are currently available to decompose multivariate Gaussian data require knowledge of the covariance matrix.  In many important problems (such as in spatial or longitudinal data analysis, and graphical modeling), the covariance matrix may be unknown and even of primary interest. Thus, in this work we develop new approaches to decompose Gaussians with unknown covariance.  First, we present a general algorithm that encompasses all previous decomposition approaches for Gaussian data as special cases, and can further handle the case of an unknown covariance.  It  yields a new and more flexible alternative to sample splitting when $n>1$.  When $n=1$, we prove that it is impossible to partition the information in a multivariate Gaussian into independent portions without knowing the covariance matrix.  Thus, we use the general algorithm to decompose  a single multivariate Gaussian with unknown covariance into \emph{dependent} parts with tractable conditional distributions, and demonstrate their use for inference and validation. The proposed decomposition strategy extends naturally to Gaussian processes. 
In simulation and on electroencephalography data, we apply these decompositions to the tasks of model selection and post-selection inference in settings where alternative strategies are unavailable. 

\end{abstract}

\section{Introduction} \label{sec:intro}

Let $X \in \real^{n \times p}$ be a random matrix with rows that are independently distributed as $N_p(\mu, \Sigma)$, where the mean vector $\mu \in \mathbb{R}^p$ and/or the $p \times p$ positive definite matrix $\Sigma$ are unknown. In the important special case of $n = 1$, $X$ is a single multivariate Gaussian random vector.

In  statistical practice, the data analyst often begins by studying $X$, and then uses the insights gained to inform downstream tasks. Examples are given in Applications \ref{ap:fit} and \ref{ap:test}. 

\begin{app}[\if\bmka1{1: }\fi Fit and validate]\label{ap:fit}
    We wish to (i) fit a model to $X$ in order to estimate the unknown parameter(s), and then (ii) validate the model by assessing its out-of-sample fit.
\end{app}

\begin{app}[\if\bmka1{2: }\fi 
Explore and confirm]\label{ap:test}
    We wish to (i) explore $X$ to generate a hypothesis involving the unknown parameter(s), and then  (ii) confirm (or reject) the hypothesis.
\end{app}

In each application, the sequential structure of the analysis is highly problematic: using the same data $X$ in both (i) and (ii) will lead to invalid inference, as pointed out by  \cite{tian2020prediction} and \cite{ oliveira2021unbiased} in the context of Application~\ref{ap:fit}, and  by \cite{fithian2014optimal} in the context of Application~\ref{ap:test}. 
A natural workaround   is to \emph{decompose} $X$ into two (or more) pieces that partition the information it contains about the unknown parameter(s). If the pieces are independent, then we can simply use the first piece for (i), and the second for (ii).

When $n>1$ and the observations are independent and identically distributed, then we can decompose $X$ using sample splitting \citep{cox1975note}:  $n_1<n$ rows form  a ``training" set used to carry out (i) in Applications~\ref{ap:fit} and \ref{ap:test}, and the remaining $n_2 = n-n_1$ observations form a ``test" set to carry out  (ii). 
However, sample splitting is unattractive or inapplicable when $n$ is small. For instance, small $n$ strictly limits the number of folds into which we can split the data: concretely, when $n < 10$, we cannot do 10-fold cross validation. At the extreme, if our data consist of a single observation ($n=1$) --- for instance, a single realization of a graph or a spatial field 
--- then sample splitting is not an option.

Other than sample splitting, what are our options for decomposing one or more Gaussians into training and test sets? Recently, it has been shown that we can  decompose a $N_p(\mu,\Sigma)$ random vector into two or more independent multivariate Gaussian vectors, provided that $\Sigma$ is known (see \citealt{rasines2021splitting, tian2018selective, oliveira2021unbiased, leiner2022data, neufeld2023data}); we will collectively refer to these proposals as ``Gaussian data thinning." However, if $\Sigma$ is unknown, then these proposals do not apply.  This is a severe limitation since there are many important settings in which $\Sigma$ is unknown and may even be the primary object of interest: examples include principal components analysis, time series analysis, spatial statistics, hierarchical models, covariance and precision graph estimation, and matrix-valued data analysis.

Our goal in this paper is to develop a unified framework to non-trivially decompose one or more realizations of a $N_p(\mu, \Sigma)$ into two or more components, when one or both of the parameters are unknown. (By ``non-trivially", we mean that each of the resulting folds depends on the parameter(s) of interest.) Towards this goal, we introduce a very simple  ``general algorithm" that is composed of two steps: first we (optionally) augment the data with ``observations" of  Gaussian noise, and then we left-multiply the augmented data with a particular matrix. It turns out that both sample splitting and Gaussian data thinning (applicable when $\Sigma$ is known) are special cases of this general algorithm. Furthermore, we show that a special case of this general algorithm yields an entirely new result: we can  decompose $n > 1$ independent Gaussians, with \emph{unknown} $\Sigma$, into $K \leq n$ independent pieces that are not simply a rearrangement of the rows of the original Gaussians; i.e. this generalizes sample splitting.

Next, we turn to the case of a single realization of a  $N_p(\mu,\Sigma)$ with unknown covariance.  This setting is commonly encountered in time series analysis, spatial statistics, and when working with matrix-valued data. None of the aforementioned decomposition strategies can be applied when $n=1$. In this setting, we prove an impossibility result: it is not possible to (non-trivially) decompose a single realization of a  $N_p(\mu,\Sigma)$ with unknown covariance into  independent pieces. Instead, we show that the general algorithm can be applied to this single realization  to obtain two or more \emph{dependent} pieces; in fact, this is a  generalization of the ``P2-data fission" proposal of \cite{leiner2022data}. This approach is fundamentally different from independent decomposition strategies: due to dependency, one cannot simply use one fold for fitting and another for validating, or one fold for exploration and another for confirmation. Rather, we must validate or confirm using \emph{conditional distributions} that account for the fact that the act of fitting or exploring inadvertently provides some information about the validation or confirmation fold. We fully address previously unexplored practical issues that arise when using these dependent folds.

Finally, we extend the ``general algorithm" and resulting decomposition strategies  to Gaussian processes. We then apply these strategies to the tasks of model selection and post-selection inference in settings where alternative strategies are unavailable. 

  Figure~\ref{fig:flowchart} displays a flowchart of the strategies for decomposing Gaussian random variables presented in this paper. 
All theoretical results are proven in the Supplement.

\begin{figure}
\centering
\includegraphics[width=1\textwidth]{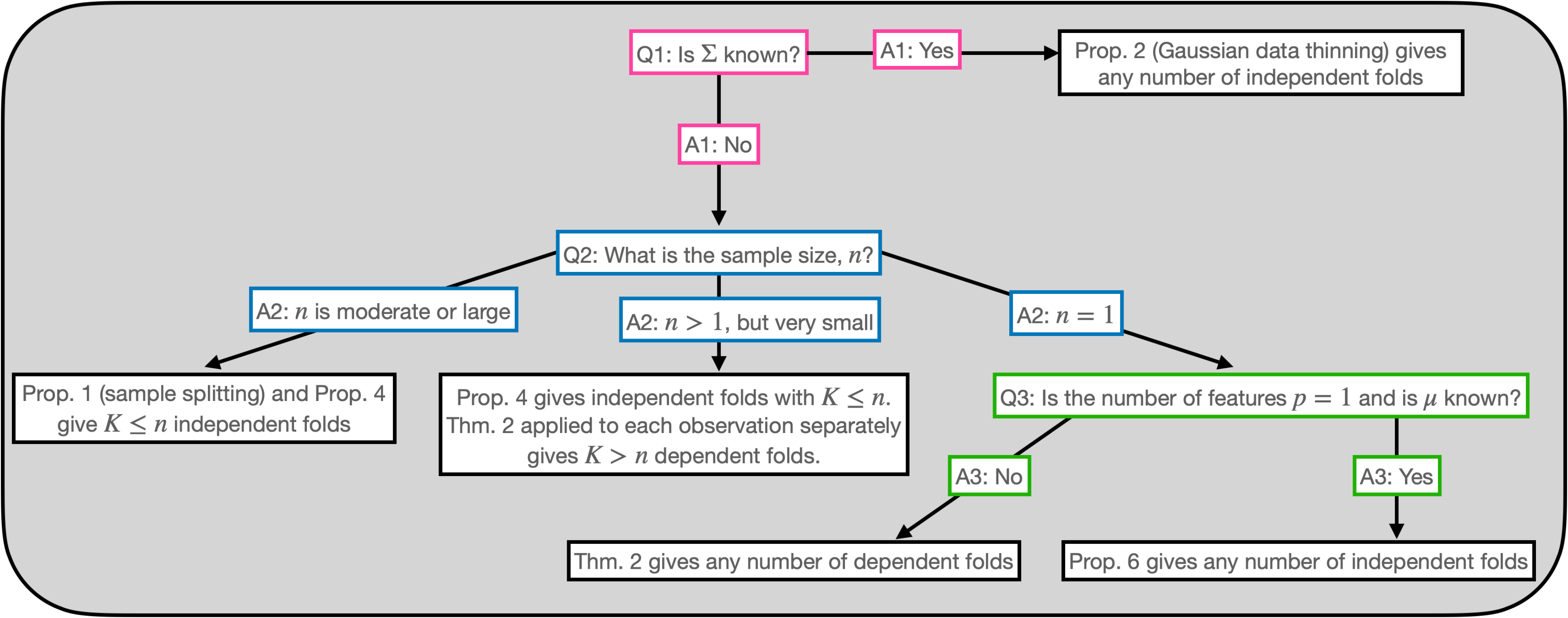}
\caption{Strategies for decomposing finite-dimensional Gaussians in Sections~\ref{sec:revisiting}--\ref{sec:fission}. All strategies arise directly from Algorithm 1 in Section~\ref{sec:premult}. Section~\ref{sec:GP} extends these results to Gaussian processes. 
}
\label{fig:flowchart}
\end{figure}

We close this section with a final remark: in this paper, we do not consider splitting the $p$ \emph{features} of $X$ into a training set with $p_1$ features and a test set with $p_2=p-p_1$ features.  While useful for validating certain interpolative or extrapolative tasks (e.g., leave-future-out cross-validation for forecasting, \citealt{burkner2020approximate}), in general there is no reason to think that a test set of $p_2$ features would provide a meaningful assessment of a model fit to the training set of $p_1$ features, absent strong structural assumptions about the unknown parameter(s). 


\section{A ``general algorithm" for decomposing independent Gaussians} 
\label{sec:premult}

Let $N_{a\times b}(\eta, \Delta, \Gamma)$ denote an $(a\times b)$-dimensional matrix-variate Gaussian with mean $\eta\in\mathbb{R}^{a \times b}$, positive definite row-covariance $\Delta\in\mathbb{R}^{a \times a}$, and positive definite column-covariance $\Gamma\in\mathbb{R}^{b \times b}$. 
 
Suppose that we are given $n$ independent realizations of a $N_p(\mu, \Sigma)$ random variable. For convenience, we will write this as $X \sim N_{n\times p}(1_n \mu^\top ,I_n,\Sigma)$.
This paper centers around a  ``general algorithm" for decomposing $X$, which we present next, and which is visually displayed in Figure~\ref{fig:alg1}.

\begin{algo}[The ``general algorithm" for decomposing $X \sim N_{n\times p}(1_n \mu^\top ,I_n,\Sigma)$ into $K$ parts]
\label{alg:premult} \textcolor{white}{.}\\
\indent \emph{Input:} a nonnegative integer $r\ge \max(K-n,0)$, a $(n+r)\times (n+r)$ orthogonal matrix $Q$, and a $p \times p$ positive definite matrix $\Sigma'$. 

\begin{enumerate}
    \item Generate $r$ independent realizations of a $N_p(0, \Sigma')$ random variable, and augment them with the rows of $X$ to obtain an $(n+r) \times p$ matrix, $X^\aug$. 
    \item Compute $X' = QX^\aug$. 
    \item Deterministically partition the $n+r$ rows of $X'$ into $K\le n+r$ submatrices, $\Xt{1},\ldots,\Xt{K}$, of dimension $n_1\times p,\ldots,n_K\times p$, respectively, where $n_1+\cdots+n_K=n+r$.  
    \end{enumerate}
    
    \indent \emph{Output:} $\Xt{1},\ldots,\Xt{K}$.

\end{algo}

\begin{figure}
\centering
\includegraphics[width=0.85\textwidth]{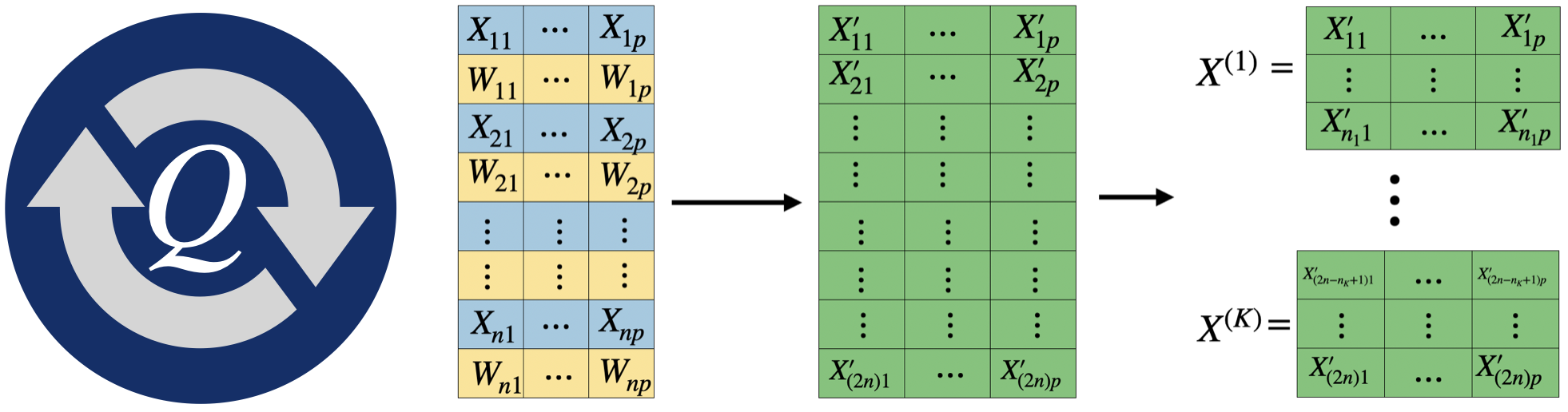}
\caption{A schematic for Algorithm 1 with $r=n$. $X_1,\ldots,X_n$ represent independent $N_p(\mu, \Sigma)$ random vectors, and $W_1,\ldots,W_n$ represent independent $N_p(0, \Sigma')$ noise vectors. $Q$ represents an $(n+r) \times (n+r)$ orthogonal matrix. }
\label{fig:alg1}
\end{figure}

\begin{remark} For convenience, we will typically  choose $r$ to be a multiple of $n$. 
Step 1 of Algorithm \ref{alg:premult} does not specify the order in which the $n$ rows of $X$ should be augmented with the $r$ noise vectors. Without loss of generality, we will construct $X^\aug$ by interspersing the rows of $X$ with the noise vectors as in Figure \ref{fig:alg1}, that is, each row of $X$ is followed by $r/n$ rows of noise. \label{remark:Xaug}
\end{remark}

\begin{remark} In Algorithm \ref{alg:premult}, since $Q$ is orthogonal, we can recover $X$ from $X'$ as follows: premultiply $X'$ by $Q^\top$ to recover $X^\aug$, and then subset the appropriate rows to form $X$. Therefore, Algorithm \ref{alg:premult} preserves all of the Fisher information about $\mu$ and $\Sigma$ contained in $X$. 
\end{remark}

\begin{remark}
\label{rem:stack}
Algorithm \ref{alg:premult} requires independent and identically distributed rows. Suppose that we instead observe $n$ non-independent or non-identically distributed Gaussians. Then, we can stack the $n$ observations into a single vector of length $np$, and apply Algorithm \ref{alg:premult} with $n=1$.
\end{remark}

The above algorithm may involve randomness either through the augmentation in Step 1 (if $r>0$) or through the matrix multiplication in Step 2 (if $Q$ is  random). The amount of information about the unknown parameters allotted to each of the $K$ submatrices is a function of the choice of $r$ and $\Sigma'$ in Step 1 (note that $\Sigma'$ need not equal the true column-covariance $\Sigma$, and indeed will necessarily be unequal if $\Sigma$ is unknown), the choice of $Q$ in Step 2, and $n_1,\dots,n_K$ in Step 3. The specifics of information allocation are discussed in detail in subsequent sections.

Though  simple, this algorithm  will serve as the foundational building block underlying all of the ideas in this paper. In particular, it will  allow us to unify all of the strategies for decomposing independent Gaussians: both existing strategies, and new strategies proposed in this paper. The choices of $Q$, $r$, and $\Sigma'$ will determine the properties of $\Xt{1},\dots,\Xt{K}$, ranging from their interdependence to the amount of Fisher information about the parameters allocated to each fold, and will have implications for their use in Applications \ref{ap:fit} and \ref{ap:test}.


\section{Revisiting recent proposals in light  of Algorithm~\ref{alg:premult}} \label{sec:revisiting}

\subsection{Recovering sample splitting from  Algorithm~\ref{alg:premult}}

Starting simple, we show that sample splitting is immediately a special case of Algorithm \ref{alg:premult}.

\begin{proposition}[Sample splitting] \label{prop:ss} Suppose that $X \sim N_{n \times p}(1_n \mu^\top ,I_n,\Sigma)$ with $n>1$, and consider Algorithm~\ref{alg:premult} with $r=0$, and $Q$  a random matrix drawn uniformly from the set of $n\times n$ permutation matrices.  
Then $\Xt{1},\ldots,\Xt{K}$ form a random partition of the rows of $X$ (uniformly over the set of all partitions of sizes $n_1,\dots,n_K$). This is exactly sample splitting.
\end{proposition}

\begin{remark}
Classical Fisher information results yield that the proportion of rows assigned to $\Xt{k}$ (i.e. $n_k/n$) is the proportion of Fisher information about the parameters allocated to $\Xt{k}$.
\end{remark}

\subsection{Recovering Gaussian data thinning from  Algorithm~\ref{alg:premult}}

When $\Sigma$ is known, many recent papers \citep{rasines2021splitting, tian2018selective, oliveira2021unbiased, leiner2022data, neufeld2023data} have considered an alternative approach to generating independent splits of Gaussian data, which we will refer to here as ``Gaussian data thinning."  This approach is especially attractive in situations when $n=1$ or $n$ is small, so that sample splitting is either unavailable or inflexible. 

We show here that Algorithm \ref{alg:premult} with a suitable choice of $Q$, $r$, and $\Sigma'$ recovers Gaussian data thinning \citep{neufeld2023data}. 

\begin{proposition}[Gaussian $K$-fold data thinning with known $\Sigma$ and $n=1$] \label{prop:thin} Suppose that $X \sim N_p(\mu,\Sigma)$ with $\Sigma$ known, and   consider Algorithm~\ref{alg:premult} with $r=K-1$,   $\Sigma'=\Sigma$, $X^\aug$ constructed according to Remark \ref{remark:Xaug}, and 
\begin{equation}
\label{eq:Qthin}
Q =  \begin{pmatrix} \sqrt{\epsilon_1} & \\
\vdots & \;\;\; U_{(\sqrt{\epsilon_1},\ldots,\sqrt{\epsilon_K} )^\top}^\perp \\
\sqrt{\epsilon_K} & \end{pmatrix},
\end{equation}
where $U_{(\sqrt{\epsilon_1},\ldots,\sqrt{\epsilon_K} )^\top}^\perp$ is a $ K\times (K-1)$ orthogonal matrix that spans  the space orthogonal to $(\sqrt{\epsilon_1},\ldots,\sqrt{\epsilon_K} )^\top$. Here,  $\epsilon_1,\ldots,\epsilon_K$ are positive scalars that sum to $1$.  

This recovers the Gaussian data thinning proposal of \cite{neufeld2023data}, in the sense that $\Xt{1},\ldots,\Xt{K}$ are independent, and marginally $\Xt{k} \sim N_p(\sqrt{\epsilon_k} \mu, \Sigma)$ for $k=1,\ldots, K$. 
\end{proposition}

\begin{remark}
\label{rem:thinfisher}
In Proposition \ref{prop:thin}, the value of $\epsilon_k$ represents the proportion of Fisher information about $\mu$ allocated to the $k$th fold. This is analogous to the role of $n_k$ in sample splitting.
\end{remark}

\begin{remark}
\label{rem:bignthin} 
Proposition \ref{prop:thin} can easily be extended to $X\sim N_{n\times p}(1_n\mu^\top, I_n,\Sigma)$ where $n>1$ and $\Sigma$ is known by applying Algorithm \ref{alg:premult} with $r=n(K-1)$, $\Sigma'=\Sigma$, and $Q=I_n\otimes Q'$, where $Q'$ is the matrix defined in \eqref{eq:Qthin}. This  is equivalent to applying Proposition \ref{prop:thin} to each row of $X$. Note that care is needed to ensure that the rows of $X^\aug$ are aligned with those of $Q$; see Remark \ref{remark:Xaug}.
\end{remark}

We can see that Proposition~\ref{prop:thin} requires knowledge of the column-covariance matrix $\Sigma$. The strategies described in the remainder of this paper do not require knowledge of $\Sigma$.

\subsection{Recovering Gaussian data fission from Algorithm \ref{alg:premult}}

\cite{leiner2022data} consider decompositions of random variables into \emph{dependent} components; they refer to such approaches as ``P2-fission". When $n=1$ and $K=2$, they propose in their supplement a P2-fission decomposition of $X \sim N_{p}( \mu,\Sigma)$ that can be applied when both $\mu$ and $\Sigma$ are unknown. As discussed in \cite{neufeld2024discussiondatafissionsplitting}, this decomposition can be thought of as a ``misspecified" version of Gaussian data thinning, where instead of adding and subtracting a mean-zero Gaussian vector with the same covariance as $X$, we instead add and subtract a mean-zero Gaussian vector with an \emph{arbitrary} covariance matrix, and subsequently  characterise the dependence between the pieces. 

The following proposition expresses \cite{leiner2022data}'s Gaussian P2-fission proposal (up to a rescaling of $1/\sqrt{2}$) as a special case of Algorithm \ref{alg:premult}. The $Q$ in Proposition \ref{prop:fission} is exactly the $Q$ in Proposition \ref{prop:thin} with $K=2$ and $\epsilon_1=\epsilon_2=1/2$. 

\begin{proposition}[Gaussian data fission with $n=1$] \label{prop:fission} Suppose that $X \sim N_p(\mu,\Sigma)$, and consider Algorithm~\ref{alg:premult} with $K=2$, $r=1$,  some positive definite matrix $\Sigma'$, $X^\aug$ constructed according to Remark \ref{remark:Xaug}, and 
$$Q =  \begin{pmatrix} {1}/{\sqrt{2}}\;\; & {1}/{\sqrt{2}} \\ {1}/{\sqrt{2}} \;\;& -{1}/{\sqrt{2}} \end{pmatrix}.$$
Take $n_1=n_2=1$. Then, the joint distribution of $\Xt{1}$ and $\Xt{2}$ is 
$$
\begin{pmatrix} \Xt{1} \\ \Xt{2} \end{pmatrix} \sim N_{2p}\left(\frac{1}{\sqrt{2}}\begin{pmatrix} \mu \\ \mu \end{pmatrix}, \frac{1}{2}\begin{pmatrix} \Sigma + \Sigma' & \Sigma - \Sigma' \\ \Sigma - \Sigma' & \Sigma + \Sigma' \end{pmatrix} \right).
$$
The marginal distribution of $\Xt{1}$ and the conditional distribution of $\Xt{2}|\Xt{1}$ can be recovered from the joint distribution using standard Gaussian manipulations.
\end{proposition}

We extend this idea to obtain $K>2$ folds, and address practical issues that arise in its application, in Section \ref{sec:fission}. 

\section{Decomposing Gaussians into independent Gaussians when $\Sigma$ is unknown}

We now show that when $n>1$, Steps 1-2 of Algorithm~\ref{alg:premult} can be applied to obtain a new decomposition of the matrix $X$ into $n$ independent $N_p(\mu, \Sigma)$ Gaussian random vectors, without knowledge of $\mu$ or $\Sigma$; these vectors can then be reconfigured in Step 3 to produce matrices $\Xt{1},\ldots,\Xt{K}$. While sample splitting can also accomplish this task, the rows of the matrices $\Xt{1},\ldots,\Xt{K}$ arising from Proposition~\ref{prop:premult} will, in general, \emph{not} be copies of the rows of $X$. Proposition~\ref{prop:premult} is illustrated in Figure~\ref{fig:prop-premult}.  

\begin{figure}
\centering
\includegraphics[width=0.85\textwidth]{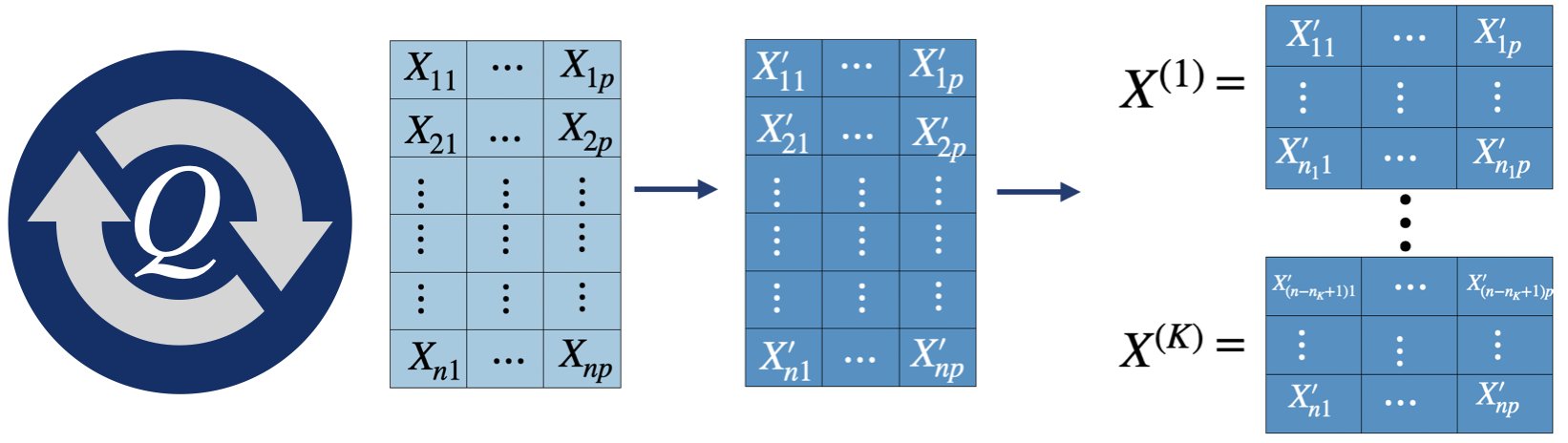}
\caption{A schematic for Proposition~\ref{prop:premult}. $Q$ represents an $n \times n$ orthogonal matrix, $X_1,\ldots,X_n$ represent independent $N_p(\mu, \Sigma)$ random vectors, and $X_1', \ldots, X_n'$ represent the rows of a  $N_{n \times p}(Q 1_n \mu^T, I_n, \Sigma)$ random matrix.}
\label{fig:prop-premult}
\end{figure}

\begin{proposition}[Decomposing $n>1$ Gaussians into $K\le n$ independent Gaussian matrices]
\label{prop:premult}
Suppose that $X \sim N_{n \times p}(1_n \mu^\top ,I_n,\Sigma)$ with  $n>1$, and  consider Algorithm~\ref{alg:premult} with $r=0$, any $K \leq n$, and 
$Q$ a $n \times n$ orthogonal  matrix.  
Then:
\begin{enumerate}
     \item $X' \sim N_{n \times p}(Q1_n \mu^\top ,I_n,\Sigma)$, i.e., the rows of $X'$ are independent realizations of a Gaussian random variable with covariance $\Sigma$. 
    \item If we further take $Q$ such that $Q1_n=1_n$,
    then $X' \sim N_{n \times p}(1_n \mu^\top, I_n, \Sigma)$, i.e. its rows are independent realizations of a $N_p(\mu, \Sigma)$ random variable.
    \end{enumerate}
\end{proposition} 
Proposition~\ref{prop:premult} does not require knowledge of $\Sigma$ (or $\mu$), and it is most useful when $\Sigma$ is unknown, since then Proposition~\ref{prop:thin} cannot be applied. 

Why should we prefer Proposition~\ref{prop:premult} to sample splitting? When $n$ is small, sample splitting is extremely inflexible: for instance, when $n=3$, there are only three ways to split $X$ into two folds. By contrast, Proposition~\ref{prop:premult} provides an infinite number of ways to do so. 

The following corollary to Proposition \ref{prop:premult} connects the proposition to the data thinning framework of \cite{neufeld2023data} and \cite{dharamshi2023generalized}. It follows from the fact that $Q$ is an $n \times n$ orthogonal matrix.

\begin{corollary}
The strategy outlined in Proposition~\ref{prop:premult} satisfies the definition of data thinning in \cite{neufeld2023data} and \cite{dharamshi2023generalized}, in the sense that (i) the rows of $X'$ are independent and depend on the unknown parameters; and (ii) they can be recombined to yield the original data (since $X=Q^\top X'$).
\end{corollary}

\begin{remark}
\label{rem:wishart}
When $\mu$ is known, there is a very close connection between the  strategy in Proposition~\ref{prop:premult}  and thinning the Wishart family. To see this, consider the case $\mu=0$. The discussion of natural exponential families in \cite{dharamshi2023generalized} suggests that it should be possible to decompose  $W=X^\top X\sim \mathrm{Wishart}_p(n,\Sigma)$ into $K$ independent $W^{(k)}\sim \mathrm{Wishart}_p(n_k,\Sigma)$ random matrices, where $n_1+\dots+n_K=n$. 
In fact,  $\left(\Xt{1} \right)^\top \Xt{1}, \ldots, \left( \Xt{K} \right)^\top \Xt{K}$ arising from Proposition~\ref{prop:premult} are independent, and follow  $\mathrm{Wishart}_p(n_k,\Sigma)$ distributions.
\end{remark}

While related to previous work, the strategy in Proposition \ref{prop:premult} is new: i.e., it was not proposed in \cite{neufeld2023data} or \cite{dharamshi2023generalized}.

The next result explains how Proposition~\ref{prop:premult} allocates Fisher information across $\Xt{1},\ldots,\Xt{K}$. 
\begin{proposition}[Allocation of Fisher information in Proposition~\ref{prop:premult}] \label{prop:premult-fisher}
Suppose that we apply Proposition \ref{prop:premult} to $X\sim N_{n\times p}(1_n\mu^\top, I_n, \Sigma)$. Let $\Qt{k}$ denote the $(n_k \times n)$-submatrix of $Q$ such that $\Xt{k}=\Qt{k}X$. 
Then, $(1_{n}^\top Q^{(k)\top} \Qt{k} 1_n)/n$ of the Fisher information about $\mu$ and $n_k/n$ of the Fisher information about $\Sigma$ is allocated to the $k$th fold. 
\end{proposition}

Proposition \ref{prop:premult} enables us to decompose $n>1$ independent realizations of a $N_p(\mu, \Sigma)$ random variable into $K \leq n$ random matrices,  consisting of $n_1,\ldots,n_K$ {\em new} independent $N_p(\mu, \Sigma)$ realizations, where $n_1+\ldots+n_K=n$. 
However, what if we want to generate more than $n$ such realizations? This is particularly critical when $n=1$, a setting that arises in  time series and spatial data applications, among others. When $p=1$ and $\mu$ is known, this is no problem: the following result follows from Example 4.3.1 of \cite{dharamshi2023generalized}. 
\begin{proposition}[Decomposing $n=1$ univariate Gaussian into $K$ independent Gammas]\label{prop:gdt}
Suppose that we observe $X\sim N(\mu,\sigma^2)$ with $\mu$ known and $\sigma^2$ unknown. Generate $Z\sim \text{Dirichlet}_K(\epsilon_1/2,\dots,\epsilon_K/2)$ where $\epsilon_1,\dots,\epsilon_K$ are positive scalars that sum to $1$, and let $(\Xt{1},\dots,\Xt{K})=(X-\mu)^2\cdot Z$. Then, $\Xt{1},\dots,\Xt{K}$ are mutually independent, and for $k=1,\dots,K$, $\Xt{k}\sim\text{Gamma}(\epsilon_k/2, 1/(2\sigma^2))$, and $\Xt{k}$ contains $\epsilon_k$ of the Fisher information about $\sigma^2$ contained in $X$.
\end{proposition}

Unlike the previous decomposition results in this paper, Proposition \ref{prop:gdt}  does not produce Gaussian random variables. 
However, it does produce independent random variables that can be used to solve Applications \ref{ap:fit} and \ref{ap:test}. Unfortunately, Proposition~\ref{prop:gdt} does not extend beyond $p=1$: our next  result reveals that for $p>1$ and general unknown $\Sigma$, \emph{it is not possible} to produce multiple independent random variables --- Gaussian or otherwise ---  from a single multivariate Gaussian. 

\begin{theorem}[Impossibility of decomposing $n=1$ multivariate Gaussian into independent pieces]\label{thm:maxn}
Suppose that $X\sim N_p(\mu,\Sigma)$ where $p>1$ and $\Sigma$ is an arbitrary covariance matrix. Absent knowledge of $\Sigma$, one cannot non-trivially decompose $X$ into multiple independent random variables. 
\end{theorem} 

Theorem~\ref{thm:maxn} tells us that we cannot use Algorithm \ref{alg:premult} (or \emph{any} algorithm, for that matter) to produce multiple independent folds in the important $n=1$ case. Briefly, its proof (i) establishes an equivalence between decomposing Gaussians with unknown covariance and decomposing the corresponding Wishart random matrix; (ii) shows that no non-additive operation that does not rely on knowledge of $\Sigma$ can recover a Wishart from independent pieces; and (iii) notes that the (singular) Wishart distribution with one degree of freedom is indecomposable \citep{SHANBHAG1976347, Peddada1991, srivastava2003singular}, i.e., it cannot be recovered from independent pieces using addition.

\section{Decomposing Gaussians into dependent Gaussians when $\Sigma$ is unknown}
\label{sec:fission}

\subsection{Generating dependent Gaussians with Algorithm \ref{alg:premult}}

Theorem~\ref{thm:maxn} established that when $\Sigma$ is unknown and $p>1$,  it is not possible to non-trivially decompose a single realization of a $N_p(\mu,\Sigma)$ random variable into multiple independent random variables. However, the $n=1$ setting is of critical importance, e.g., in the context of spatial or temporal data. We will now show that Algorithm~\ref{alg:premult} can be applied to decompose a single realization of a $N_p(\mu,\Sigma)$ random variable into any number of \emph{dependent} Gaussians. 

\begin{theorem}[Decomposing $n=1$ Gaussian into $K$ dependent Gaussians] \label{thm:dependent}
Suppose that $X \sim N_p( \mu,\Sigma)$, and  consider Algorithm~\ref{alg:premult} with $r \in \mathbb{Z}_+$, $K= r+1$, $\Sigma'$ a positive definite matrix selected by the data analyst, $X^\aug$ constructed so that $X$ is in its first row, and $Q$ an orthogonal matrix of dimension $K\times K$. Let $Q_X$ denote the first column of $Q$.  
Then, marginally, 
$$
\text{vec}(X'^\top) = 
    \begin{pmatrix} \Xt{1} \\ \vdots \\ \Xt{K} \end{pmatrix}
    \sim  N_{Kp}\left(\text{vec}(\mu Q_X^\top),  Q_XQ_X^\top \otimes \Sigma + (I_K-Q_XQ_X^\top) \otimes \Sigma'
    \right).
$$
\end{theorem}

\begin{remark} \label{remark:conditional}
In the special case where $r=1$ and $K=2$, letting $q_1$ and $q_2$ denote the first and second entries of $Q_X$ respectively, it follows that $\Xt{1} \sim N_p(q_1\mu,q_1^2\Sigma+(1-q_1^2)\Sigma')$ and 
\begin{align*}
\left[\Xt{2} \mid \right.&\left.\Xt{1}=\xt{1}\right] \sim N_p\left(q_2\mu+q_1q_2\left(\Sigma-\Sigma'\right)\left(q_1^2\Sigma+\left(1-q_1^2\right)\Sigma'\right)^{-1}\left(\xt{1}-q_1\mu\right)\right., \\
&\left.q_2^2\Sigma+\left(1-q_2^2\right)\Sigma'-q_1^2q_2^2\left(\Sigma-\Sigma'\right)\left(q_1^2\Sigma+\left(1-q_1^2\right)\Sigma'\right)^{-1}\left(\Sigma-\Sigma'\right)\right).
\end{align*}
Furthermore, if $Q_X=\begin{pmatrix} 1/\sqrt{2} \;\;\; & 1/\sqrt{2}\end{pmatrix}^\top$, then Theorem~\ref{thm:dependent} yields Proposition \ref{prop:fission}. 
\end{remark} 
Figure~\ref{fig:thm:dependent} provides an illustration of Theorem~\ref{thm:dependent}. 

\begin{figure}
\centering
\includegraphics[width=0.85\textwidth]{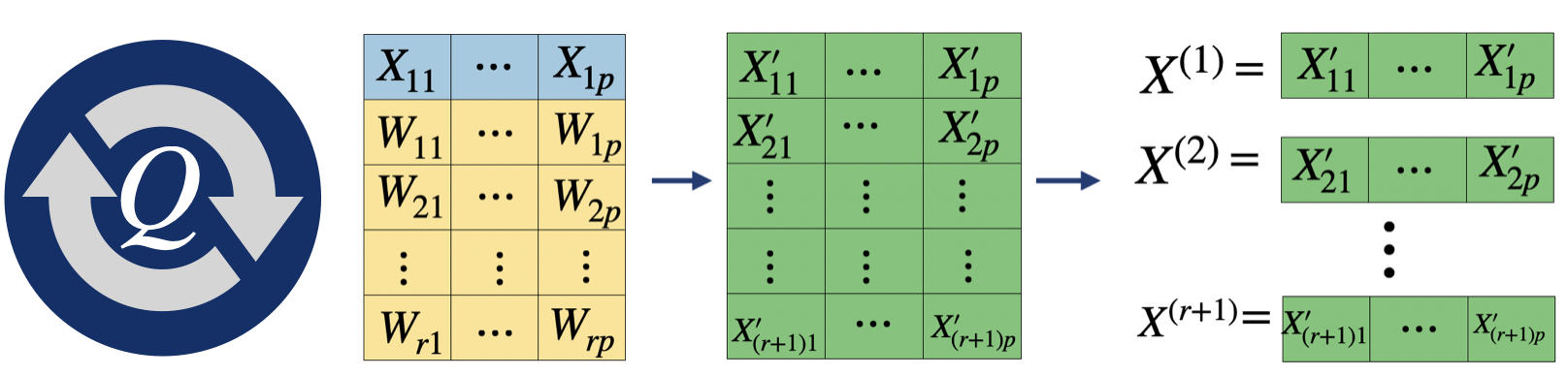}
\caption{A schematic for Theorem~\ref{thm:dependent} with $n=1$ and $K=r+1$. $Q$ represents an $n \times n$ orthogonal matrix, $X$ represent a $N_p(\mu, \Sigma)$ random vector, and $W_1, \ldots, W_r$ represent independent $N_p(0, \Sigma')$ random vectors.}
\label{fig:thm:dependent}
\end{figure}

Applying Theorem~\ref{thm:dependent} in  a practical setting incurs a number of challenges due to the dependence between $\Xt{1},\ldots,\Xt{K}$.
For instance, suppose that $K=2$: we cannot simply fit a model using $\Xt{1}$, and then validate it using $\Xt{2}$. Instead,
\cite{leiner2022data} propose to fit a model using $\Xt{1}$, and to validate it using the \emph{conditional distribution}  $\Xt{2}|\Xt{1}$. However, details of how this can be done in practice are not addressed. 

In the remainder of this section, we address the practical questions that arise in applying Theorem~\ref{thm:dependent}.  
We contextualise Theorem \ref{thm:dependent} in terms of Applications~\ref{ap:fit} and \ref{ap:test}, discuss how the Fisher information is allocated between $\Xt{1},\ldots,\Xt{K}$, and consider the choice of $\Sigma'$. 

In this section, we consider decomposing a single Gaussian into dependent folds, because when 
$n>1$, it is possible to obtain $K\le n$ \emph{independent} folds via Proposition~\ref{prop:premult} --- and of course, independent folds are more convenient for downstream analysis.  
However, in the event that $K>n$ folds are desired (for example, for cross-validation when $n$ is small), a direct application of Algorithm~\ref{alg:premult} with $r>0$ decomposes $n>1$ independent Gaussians into $K>n$ dependent folds, necessitating a slight extension of Theorem~\ref{thm:dependent}; see Supplement \ref{app:thm:dependent:nbig} for details. We recommend using  a block diagonal $Q$ matrix, as in Remark \ref{rem:bignthin}; this  amounts to applying Theorem~\ref{thm:dependent} to each row of $X$ separately.

\subsection{Revisiting Application~\ref{ap:fit} with $n=1$} \label{subsec:application-1}

In this subsection, we consider applying Theorem~\ref{thm:dependent} in the context of Application~\ref{ap:fit} with $n=1$.  We will address two issues: (1) how to handle dependence between training and test sets when $K=2$; and (2) how to handle the case of $K>2$.

First, we consider the dependence between training and test sets in the case where $K=2$. 
We can easily carry out step (i) of Application~\ref{ap:fit} by 
fitting  a model to $\Xt{1}$. But step (ii) poses a challenge: dependence between $\Xt{1}$  and $\Xt{2}$ means that we cannot simply assess the model's out-of-sample fit using the marginal distribution of $\Xt{2}$. Instead, we must conduct step (ii)  using the conditional distribution of $\Xt{2} \mid \Xt{1}$. The use of this conditional distribution in step (ii) ``accounts" for the use of $\Xt{1}$ in step (i). 
The details are as follows.

\begin{example}[Solving Application \ref{ap:fit} when $n=1$ using $K=2$ dependent folds]\label{ex:validate}
Suppose we are given a single realization of a random vector $X\sim N_p(\mu,\Sigma)$. To fit a model and assess its out-of-sample fit, we can take the following approach: 
\begin{enumerate}
    \item Apply Algorithm~\ref{alg:premult} with $r=1$, $Q$ a $2\times 2$ orthogonal matrix, and $K=2$ to decompose $X$ into $\Xt{1}$ and $\Xt{2}$. 
    \item Let $\hat\mu$ and $\hat\Sigma$ denote estimates of $\mu$ and $\Sigma$ using $\Xt{1}$, which has marginal distribution $\Xt{1}\sim N_p(q_1\mu,q_1^2\Sigma + (1-q_1^2)\Sigma')$ by Remark \ref{remark:conditional}. 
    \item Evaluate the conditional log-likelihood $\mathcal{L}(\hat\mu, \hat\Sigma;\Xt{2}|\Xt{1})$ using the conditional distribution of $\Xt{2} \mid \Xt{1}$ given in Remark~\ref{remark:conditional}. 
\end{enumerate}
\end{example}

Next, we suppose again that $n=1$, but  now $K>2$. For $k=1,\ldots, K$, we wish to fit a model to $\Xt{1},\ldots,\Xt{k-1}, \Xt{k+1}\,\ldots,\Xt{K}$, and validate it on $\Xt{k}$. However, there is a problem: it may not be straightforward or desirable to apply a model designed for a single Gaussian vector, $X$, to $K-1$ dependent Gaussian vectors. 
The following corollary enables us to  collapse subsets of $\Xt{1},\dots,\Xt{K}$ into vectors, without losing information about $\mu$ or $\Sigma$. 

\begin{corollary}
\label{cor:collapse}
Consider the setting of Theorem \ref{thm:dependent} where $n=1$. Suppose that $A$ and $B$ form a non-overlapping subset of $\{1,\ldots K\}$: that is, $A\cup B \subseteq \{1,\ldots,K \}$ and $A\cap B = \emptyset$. 
 Let $Q_A$ be a length $K$ vector where the $k$th entry equals the $k$th entry of $Q_X$ if $k\in A$ and $0$ otherwise, and define $Q_B$ similarly. 
Define $\Xt{A}=X'^\top Q_A$ and $\Xt{B}=X'^\top Q_B$. For notational convenience, let $d_A = Q_A^\top Q_A$ and $d_B = Q_B^\top Q_B$. Then,
\begin{align*}
\begin{pmatrix} \Xt{A} \\ \Xt{B} \end{pmatrix} &\sim N_{2p} \left(
\begin{pmatrix} d_A\mu  \\ d_B\mu \end{pmatrix}, 
\begin{pmatrix}
d_A^2\Sigma + d_A(1-d_A)\Sigma' &
d_Ad_B(\Sigma -\Sigma')  \\
d_Ad_B(\Sigma -\Sigma') &
d_B^2\Sigma + d_B(1-d_B) \Sigma'
\end{pmatrix}
\right), \;\; and 
\end{align*}
\begin{align*}
\left[\Xt{B}|\right.&\left.\Xt{A}=\xt{A}\right]\sim N_p\left(d_B\mu + d_B\left(\Sigma -\Sigma'\right)\left(d_A\Sigma + (1-d_A)\Sigma'\right)^{-1}\left(\xt{A}-d_A\mu\right),\right. \\
&\left.d_B^2\Sigma + d_B(1-d_B)\Sigma' - d_Ad_B^2\left(\Sigma -\Sigma'\right)\left(d_A\Sigma + (1-d_A)\Sigma'\right)^{-1}\left(\Sigma -\Sigma'\right)\right). 
\end{align*} 

Furthermore, since $\Xt{A} + \Xt{B} = X$, all of the Fisher information about $\mu$ and $\Sigma$ in $X$ is retained in $\Xt{A}$ and $\Xt{B}$.
\end{corollary}

Thus, rather than fitting a model to  $\Xt{1},\ldots,\Xt{k-1}, \Xt{k+1}\,\ldots,\Xt{K}$ --- a task that may be challenging if the model is designed for a single vector --- Corollary~\ref{cor:collapse} enables us to collapse these $K-1$ vectors into a single vector. We can then fit a model to this single vector. 
 Thus, Example~\ref{ex:validate} can be easily extended to allow for cross-validation.

\subsection{Revisiting Application~\ref{ap:test} with $n=1$}
\label{subsec:application-2}

We now briefly consider Application~\ref{ap:test}. Taking $K=2$, we simply select a hypothesis using $\Xt{1}$, and test it using $\Xt{2} \mid \Xt{1}$.

\begin{example}[Solving Application \ref{ap:test} when $n=1$  using $K=2$ dependent folds] 
\label{ex:test}
Suppose we are given a single realization of a random vector $X\sim N_p(\mu,\Sigma)$. To generate a hypothesis involving $\mu$ and/or $\Sigma$ on the data, and then to confirm (or reject) it on the same data, we can take the following approach:
\begin{enumerate}
    \item Apply Algorithm~\ref{alg:premult} with $r=1$, $Q$ a $2\times 2$ orthogonal matrix, and $K=2$ to decompose $X$ into $\Xt{1}$ and $\Xt{2}$. 
    \item Select a hypothesis involving $\mu$ and/or $\Sigma$ using the data $\Xt{1}$, which has marginal distribution $\Xt{1}\sim N_p(q_1\mu,q_1^2\Sigma + (1-q_1^2)\Sigma')$ by Remark \ref{remark:conditional}. 
    \item Test the selected hypothesis  using the conditional distribution of $\Xt{2}|\Xt{1}$ given in Remark~\ref{remark:conditional}. The details of this test are context-specific.
\end{enumerate}
\end{example}

\subsection{Allocation of Fisher information  across folds}
\label{subsec:fission-fisher}
  
The next result specifies the allocation of Fisher information in $\Xt{1}$ versus in $\Xt{2} \mid \Xt{1}$. 

\begin{theorem}[Fisher information allocation from Corollary \ref{cor:collapse}]
\label{thm:fisher} %
Suppose that we apply Algorithm~\ref{alg:premult} to $X\sim N_{n\times p}(1_n\mu(\theta)^\top,I_n,\Sigma(\phi))$ with $n=1$, $r=1$, $K=2$, $\Sigma'$ a positive definite matrix selected by the data analyst, and $Q$ an orthogonal matrix of dimension $2\times 2$, where $\theta$ is the unknown parameter vector of length $M \le p$ and $\phi$ is the unknown parameter vector of length $N \le p(p+1)/2$ that characterise the mean and covariance, respectively. 

Then, letting $q_1$ denote the top-left element of $Q$, the Fisher information about $\theta$ and $\phi$ in $\Xt{1}$ and $\Xt{2}|\Xt{1}$ is as follows: 
\begin{align*}
    \left[I_{\Xt{1}}(\theta)\right]_{ii'} &= q_1^2\frac{\partial \mu(\theta)^\top}{\partial \theta_i}\left(q_1^2\Sigma(\phi)+(1-q_1^2)\Sigma'\right)^{-1}\frac{\partial \mu(\theta)}{\partial \theta_{i'}} \\ 
    \left[I_{\Xt{2}|\Xt{1}}(\theta)\right]_{ii'} &= \frac{\partial \mu(\theta)^\top}{\partial \theta_i}\Sigma(\phi)^{-1}\frac{\partial \mu(\theta)}{\partial \theta_{i'}} - \left[I_{\Xt{1}}(\theta)\right]_{ii'}, \\
    \\
    \left[I_{\Xt{1}}(\phi)\right]_{jj'} &= \frac{1}{2}q_1^4 \trace\left[\left(q_1^2\Sigma(\phi)+(1-q_1^2)\Sigma'\right)^{-1} \frac{\partial \Sigma(\phi)}{\partial \phi_j} \left(q_1^2\Sigma(\phi)+(1-q_1^2)\Sigma'\right)^{-1} \frac{\partial \Sigma(\phi)}{\partial \phi_{j'}} \right], \\
    \left[I_{\Xt{2}|\Xt{1}}(\phi)\right]_{jj'} &= \frac{1}{2}\trace\left[\Sigma(\phi)^{-1}\frac{\partial \Sigma(\phi)}{\partial \phi_j} \Sigma(\phi)^{-1}\frac{\partial \Sigma(\phi)}{\partial \phi_{j'}} \right] - \left[I_{\Xt{1}}(\phi)\right]_{jj'}, 
\end{align*}
where $1\le i,i' \le M$, and $1\le j,j' \le N$. 
\end{theorem}

It is clear from Theorem~\ref{thm:fisher} that the amount of Fisher information allocated to $\Xt{1}$ versus  $\Xt{2} \mid \Xt{1}$  is a complicated function of $Q$ and $\Sigma'$, and depends on the unknown parameters. Thus, it cannot be finely controlled.  However, careful choices of $Q$ and $\Sigma'$ can produce desirable allocations: for example, in the setting of Theorem \ref{thm:fisher}, if $q_1^2=1/2$, then fitting a model to $\Xt{1}$ and validating it on $\Xt{2}|\Xt{1}$ leads to an equivalent allocation of information as fitting a model to $\Xt{2}$ and validating it on $\Xt{1}|\Xt{2}$. The following subsection, and Supplement \ref{app:tuning}, offer guidelines for selecting these hyperparameters.

\subsection{Computational considerations in the choice of $\Sigma'$} \label{subsec:fission-computation}

Step 1 of Examples~\ref{ex:validate} and \ref{ex:test} involves Algorithm~\ref{alg:premult}, which requires the user to choose a $p \times p$ positive definite matrix $\Sigma'$. We will now see that the choice of $\Sigma'$ has important computational implications for Steps 2 and 3: in short, we recommend choosing it to be a multiple of the identity.

Step 2 of Examples~\ref{ex:validate} and \ref{ex:test} involves estimating $\Sigma$ from $\Xt{1} \sim N_p(q_1\mu,q_1^2\Sigma+(1-q_1^2)\Sigma')$. For an arbitrary $\Sigma'$, this may be challenging: a poorly chosen $\Sigma'$ may disrupt the structure of the covariance model such that fitting a model to $\Xt{1}$ is substantially more challenging or computationally expensive than fitting a model to $X$. For instance, if $X$ follows an autoregressive model, then $\Xt{1}$ will not in general follow an autoregressive model, and therefore, standard tools for fitting autoregressive models cannot be applied to $\Xt{1}$.

Fortunately, in the special case that $\Sigma' = \text{diag}(\sigma_1'^{2},\dots,\sigma_p'^2)$,  where $\sigma_j'^2>0$ are chosen by the data analyst, the situation simplifies: in this case,
\begin{equation}
    \Xt{1}\sim N_p\left(q_1\mu,q_1^2\Sigma + (1-q_1^2)\text{diag}(\sigma_1'^2,\dots,\sigma_p'^2)\right). \label{eq:lgm1}
    \end{equation}
In fact, this can be re-written as a latent Gaussian model  with an independent observation layer: 
\begin{equation}
    \Xt{1}_j|Y \sim N\left(q_1 Y_j, (1-q_1^2)\sigma_j'^2\right),\;\;  j=1,\dots,p; \quad\quad\quad
    Y \sim N_p(\mu, \Sigma). \label{eq:lgm2}
\end{equation}
Latent Gaussian models of the form in \eqref{eq:lgm2} are a well-studied class of problems (see \citealt{rueheld, durbin2012time}), and fast algorithms are available to fit these models (e.g., the integrated nested Laplace approximation  method of \citealt{rue2009approximate}). Thus, if $\Sigma'$ is chosen to be diagonal, then it is straightforward to extend a model for $X$ to $\Xt{1}$.  

Furthermore, suppose that our model involves structural constraints on the unknown matrix $\Sigma$ under which an eigendecomposition can be efficiently computed; an autoregressive model is one such example. Then, setting 
$\Sigma' = \sigma'^2I_p$ may lead to substantial additional computational benefits in computing the likelihoods of $\Xt{1}$ \emph{and} $\Xt{2}| \Xt{1}$ with respect to candidate values for $\Sigma$ (here denoted $\Sigma^*$). To see this, note that if we let  $\Sigma^*=P\Lambda P^\top$ denote the eigendecomposition of a  candidate value  $\Sigma^*$, then 
\begin{align}
P^\top\Xt{1} \sim N_p(&q_1P^\top\mu,q_1^2\Lambda + (1-q_1^2)\sigma'^2I_p), \label{eq:fastx1} \\
\left[P^\top\Xt{2}|\Xt{1}=\xt{1}\right] \sim N_p(&q_2P^\top\mu + q_1q_2(\Lambda-\sigma'^2I_p)(q_1^2\Lambda + (1-q_1^2)\sigma'^2I_p)^{-1}P^\top(\xt{1}-q_1\mu),\nonumber \\
q_2^2\Lambda+(1-q_2^2)\sigma'^2I_p&-q_1^2q_2^2(\Lambda-\sigma'^2I_p)(q_1^2\Lambda+(1-q_1^2)\sigma'^2I_p)^{-1}(\Lambda-\sigma'^2I_p)). \label{eq:fastx2x1}
\end{align} 
Since the covariances in \eqref{eq:fastx1} and \eqref{eq:fastx2x1} are diagonal by construction, all required likelihood evaluations can 
be performed with univariate Gaussian computations.  
As we will see in Section \ref{sec:sims}, the reduction given in \eqref{eq:fastx2x1} is particularly relevant for Example \ref{ex:test}, as constructing a test statistic often requires numerically optimizing the likelihood of $\Xt{2}|\Xt{1}$ with respect to $\Sigma^*$, which may otherwise be quite challenging.

We note that \eqref{eq:lgm2} holds for any diagonal choice of $\Sigma'$, and that \eqref{eq:fastx1} and \eqref{eq:fastx2x1} hold for any $\Sigma'$ that is a positive multiple of the identity matrix. The specific values on the diagonal have implications for the allocation of Fisher information between the dependent folds, as the equations in Theorem \ref{thm:fisher} are functions of $\Sigma'$. We discuss these implications in greater detail in Supplement~\ref{app:tuning}. 

\section{Extending Algorithm~\ref{alg:premult} to Gaussian processes}
\label{sec:GP}

In this section, we extend Algorithm~\ref{alg:premult} to the case of \emph{infinite}-dimensional Gaussian processes, such as arise in the context of probabilistic machine learning, Gaussian process regression, and continuously-indexed spatial data \citep{williams2006gaussian, schabenberger2017statistical}. To do this, we simply substitute the Gaussian noise vectors in Algorithm~\ref{alg:premult} for Gaussian noise processes. We begin with some notation. We will write $GP(\mu, C)$ to refer to the Gaussian process with mean function $\mu(t)$ and covariance function $C(t,t')$, where  $t,t'\in T$ for an  index set $T$. Here $\mu(t)$ and $C(t,t')$ are unknown functions. Our data take the form of $n$ independent Gaussian processes, $X_1(t),\ldots,X_n(t) \sim GP(\mu, C)$. 

\begin{algo}[The ``general algorithm" for decomposing $X_1(t),\ldots,X_n(t) \sim GP(\mu, C)$ into $K$ parts]
\label{alg:premult-GP} \textcolor{white}{.}\\

\indent \emph{Input:} a positive integer $K$, a nonnegative integer $r$, a $(n+r)\times (n+r)$ orthogonal matrix $Q$, and a positive definite covariance function $C'(t,t')$ for all $t, t' \in T$. 

    \begin{enumerate}
    \item Generate $W_1(t),\ldots,W_r(t) \sim GP(0,C')$, i.e.,  $r$ independent realizations of a $GP(0,C')$ Gaussian process.
    \item For $i=1,\ldots,n+r$, construct the Gaussian process $X'_i(t) = \sum_{j=1}^{n} Q_{ij}X_j(t) + \sum_{j=1}^r Q_{i(n+j)} W_j(t)$,  where $Q_{ij}$ refers to the $ij$th entry of $Q$. 
    \item Deterministically partition $X'_1(t),\ldots,X'_{n+r}(t)$ into  $K \le n+r$ subsets, $\Xt{1}(t),\dots,\Xt{K}(t)$, of size $n_k$, respectively, where $n_1+\dots+n_K=n+r$.
    \end{enumerate}
    
    \indent \emph{Output:} $\Xt{1}(t),\dots,\Xt{K}(t)$.
\end{algo}

All of the finite-dimensional Gaussian decomposition strategies in the flowchart in Figure~\ref{fig:flowchart} extend naturally to the setting of Gaussian processes.
The remainder of this section focuses on the case where $n=1$, as this is the setting that most
commonly arises. 

In the event that $C$ is known, the next proposition, which mirrors Proposition \ref{prop:premult} in the finite-dimensional setting, shows that one can decompose a single Gaussian process into $K$ independent Gaussian processes.

\begin{proposition}[Decomposing $n=1$ Gaussian process into $K$ independent Gaussian processes] \label{prop:GPthin} Suppose that $X(t) \sim GP(\mu,C)$ with $C$ known, and consider Algorithm~\ref{alg:premult-GP} with $r\in\mathbb{Z}_+$, $K=r+1$, $C'=C$, and $Q$ the $(r+1)\times(r+1)$ orthogonal matrix given in \eqref{eq:Qthin}. Then, $\Xt{1}(t),\dots,\Xt{K}(t)$ are independent Gaussian processes, and marginally, $\Xt{k}(t)\sim GP(\sqrt{\epsilon_k}\mu,C)$ for $k=1,\dots,K$.
\end{proposition}

The next result mirrors Theorem~\ref{thm:dependent} and Remark \ref{remark:conditional}: that is, it enables us to decompose a single Gaussian process into $K$ dependent Gaussian processes, and to characterize their dependence. 

\begin{theorem}[Decomposing $n=1$ Gaussian process into $K$ dependent Gaussian processes]\label{thm:dependent-GP}
Suppose that $X(t)\sim GP(\mu,C)$, and consider Algorithm \ref{alg:premult-GP} with $r\in\mathbb{Z}_+$, $K=r+1$, $C'$ a positive definite covariance function selected by the data analyst, and $Q$ an orthogonal matrix of dimension $(r+1)\times(r+1)$. Let $Q_X$ denote the first column of $Q$. Let $\Xt{1}(t),\dots,\Xt{K}(t)$ denote the output of  Algorithm~\ref{alg:premult-GP}. Then, each $\Xt{k}(t)$ is marginally a Gaussian process,
$$
\Xt{k}(t) \sim GP(q_k\mu, q_k^2C+(1-q_k^2)C'),
$$
where $q_k$ is the $k$th entry of $Q_X$.

Further, suppose that $K=2$. Then, for any finite index set $T_d=\{t_1,\dots,t_d\}\in T$, the conditional distribution of $(\Xt{2}(t_1),\dots,\Xt{2}(t_d))$ given $(\Xt{1}(t_1),\dots,\Xt{1}(t_d))$  is 
\begin{align}
&N_d\left(q_2\mu_{T_d} + q_1q_2\left(\Sigma_{T_d}-\Sigma'_{T_d}\right)\left(q_1^2\Sigma_{T_d} + \left(1-q_1^2\right)\Sigma'_{T_d}\right)^{-1}\left(\left(\xt{1}_{t_1},\dots,\xt{1}_{t_d}\right)-q_1\mu_{T_d}\right)\right., 
\nonumber \\
&\left.q_2^2\Sigma_{T_d}+\left(1-q_2^2\right)\Sigma'_{T_d}-q_1^2q_2^2\left(\Sigma_{T_d}-\Sigma'_{T_d}\right)\left(q_1^2\Sigma_{T_d} + \left(1-q_1^2\right)\Sigma'_{T_d}\right)^{-1}\left(\Sigma_{T_d} - \Sigma'_{T_d}\right)\right),\label{eq:GP-K2}
\end{align}
where $\mu_{T_d}$ is the vector constructed by evaluating $\mu(t)$ at every point in $T_d$, and $\Sigma_{T_d}$ and $\Sigma'_{T_d}$ are the covariance matrices constructed by evaluating $C(t,t')$ and $C'(t,t')$, respectively, at every pair of points in $T_d$.
\end{theorem}

Knowledge of the conditional distribution in \eqref{eq:GP-K2} enables us to employ the dependent Gaussian processes arising from Theorem~\ref{thm:dependent-GP} for model evaluation and inference tasks.

In Theorem~\ref{thm:dependent-GP}, while any covariance function $C'$ could be used, the discussion in Section \ref{subsec:fission-computation} suggests using a white noise process: that is, $C'(t,t')=\sigma'^2\delta_{t=t'}$ for all $t,t'\in T$, where $\delta_{t=t'}$ is an indicator variable that equals $1$ if $t=t'$ and $0$ otherwise, and $\sigma'$ a positive constant.

We can extend Examples~\ref{ex:validate} and \ref{ex:test} to the setting of Gaussian processes using the results in this section; details are omitted due to space constraints.

\section{Illustrative examples}
\label{sec:sims}

\subsection{Context}

In this section, we illustrate the use of Algorithm \ref{alg:premult} to solve Applications \ref{ap:fit} and \ref{ap:test}. We are motivated by the analysis of electroencephalograpy data, which presents as a matrix of readings, $X$, where each row of $X$ corresponds to an electrode and each column corresponds to a time point. \cite{zhou2014gemini} models $X$ as a single realization of a matrix-variate Gaussian distribution, $X \sim {N}_{a\times b}(\mathbf{\mu}, \Delta, \Gamma)$, where $a$ is the number of electrodes, $b$ is the number of time points, the row-covariance of $X$ is $\Delta \in \mathbb{S}^{a}$, and the column-covariance of $X$ is $\Gamma \in \mathbb{S}^{b}$; here, $\mathbb{S}^{d}$ denotes the cone of $d\times d$ positive semi-definite matrices. 

Our interest lies in the row-covariance $\Delta$, which encodes the relationship between the electrodes. Specifically, we wish to test whether certain elements of $\Delta$ selected using the data are equal to $0$, or to fit and validate a clustering model for the electrodes. 
These tasks are instances of Applications \ref{ap:test} and \ref{ap:fit}, respectively. Due to temporal dependence, one cannot simply treat the columns as independent Gaussian vectors, and thus sample splitting is not a viable solution. Rather, following Remark \ref{rem:stack}, we will stack the columns of $X$ into a single vector and will then apply Algorithm \ref{alg:premult} to $\text{vec}(X)\sim N_{ab}(\text{vec}(\mu),\Gamma\otimes\Delta)$ to decompose it into dependent folds.

For the remainder of this section, we will assume that $\mu=\mathbf{0}_{a,b}$, and take $\Gamma$ to be a first-order autoregressive covariance matrix with unknown parameter $\rho$, denoted by $\Gamma(\rho)$. 
The covariance of $\text{vec}(X)$ admits a fast eigendecomposition, and thus we will use the strategy outlined in Section \ref{subsec:fission-computation} for all likelihood evaluations; see Supplement \ref{app:fast} for details. Code to reproduce the analysis in this section is available at \url{https://github.com/AmeerD/Gaussians/}.

\subsection{Simulation study}
\label{subsec:ex-test}

We illustrate the use of Algorithm \ref{alg:premult} using simulated data with $a=25$ and $b=100$ in the context of Application \ref{ap:test}. As our selection procedure, we will identify the largest absolute off-diagonal entry in $\hat\Delta$, the sample-based estimate of $\Delta$, and will select the null hypothesis that the corresponding element of $\Delta$ is equal to $0$. As this hypothesis is data-driven, methods that reuse information used in selection for inference will be highly problematic. 

We will consider three methods to accomplish this task: method (a), a naive approach that uses $X$ to select the largest entry of the sample row-covariance matrix, and again uses $X$ to test if that entry is zero; method (b), an approach that selects the largest entry of the sample row-covariance of $\Xt{1}$ from Algorithm \ref{alg:premult}, and uses the marginal distribution of $\Xt{2}$ to test if that entry is zero (i.e. it ignores the dependence between folds when testing); and method (c), an approach that selects the largest entry of the sample row-covariance of $\Xt{1}$ from Algorithm \ref{alg:premult}, and uses the conditional distribution of $\Xt{2}|\Xt{1}$ to test if that entry is zero (i.e. it accounts for the dependence between folds when testing). Full details on the three methods are given in Supplement \ref{app:sims-test}. We note that in methods (b) and (c), the covariance structures of $X$ and $\Xt{1}$ will be identical so long as a diagonal $\Sigma'$ is used in Algorithm \ref{alg:premult}, so that $\Xt{1}$ can be easily used for selection. 

We consider two settings for the unknown parameters. In the first ``null" setting, we set $\rho=0.9$ and $\Delta=I_{25}$. As there is no true signal, p-values should be uniform. However, we will see that due to the recycling of information between selection and inference, methods (a) and (b) will fail, and only method (c) will successfully control the Type I error rate. In the second ``alternative" setting, $\rho=0.9$ again, but $\Delta=\text{diag}(\Delta_1,\Delta_2)$ is a block diagonal matrix where $\Delta_1=\omega\cdot1_21_2^\top+(1-\omega)\cdot I_2$ for $\omega\in\{-0.9,-0.8,\dots,0.8,0.9\}$ and $\Delta_2=I_{23}$. That is, the first and second coordinates (electrodes) of $X$ have non-zero covariance. Our objective in this setting is to study the power of method (c).

For the ``null" setting, we first generate one sample of $X\sim N_{25\times 100}(\mathbf{0}_{25\times100},I_{25},\Gamma(0.9))$ and then apply each of the three methods to select and test an entry of the row-covariance. For the ``alternative" setting, we generate one sample of $X\sim N_{25\times 100}(\mathbf{0}_{25\times100},\text{diag}(\Delta_1,\Delta_2),\Gamma(0.9))$ for each value of $\omega$ and then apply method (c). We repeat these processes 1000 times.

The results of this simulation study are given in Figure \ref{fig:inf}. In Panel \ref{fig:t1e}, we display the ``null" setting p-values under each method against uniform quantiles. As expected, only method (c) controls the Type I error rate when there is no true row-covariance in the data. For the alternative setting, our interest is in the power of method (c): that is, the probability of rejecting the null hypothesis that $\delta_{12}$, the covariance between the first two coordinates of $X$, is zero. As the selection step is not guaranteed to find $\delta_{12}$, following \cite{gao2020selective}, we first plot the ``detection probability", i.e. the proportion of replicates that select $\delta_{12}$, as a function of $\omega$ in Panel \ref{fig:detect} for three choices of $Q$. We then plot the conditional power as a function of $\omega$ in Panel \ref{fig:cpow}; this is the probability of rejecting the null hypothesis that $\delta_{12}=0$, given that $\delta_{12}$ was selected. Together, Panels \ref{fig:detect} and \ref{fig:cpow} show that as $|\omega|$ increases, both detection and power improve. Moreover, as $q_1$ increases, the increased allocation of Fisher information to $\Xt{1}$ leads to improved detection, though at the expense of power, as less information is left in $\Xt{2}|\Xt{1}$.

\begin{figure}
\centering
\begin{subfigure}{0.33\textwidth}
    \centering
    \includegraphics[width=\textwidth]{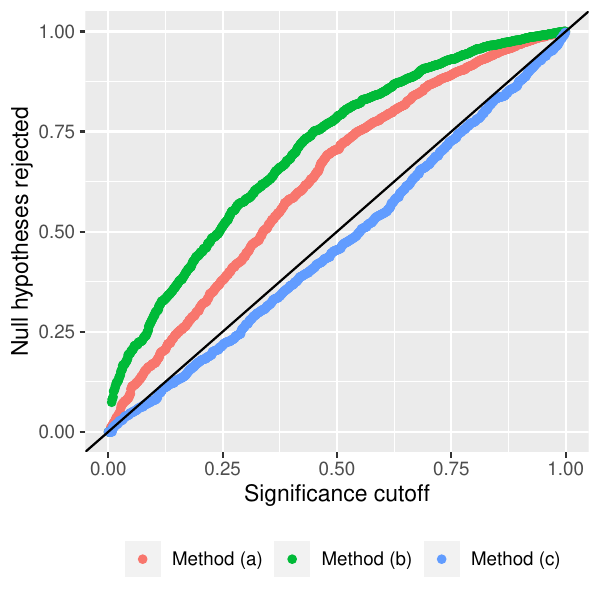}
    \caption{Type 1 error \if\bmka1{quantile-quantile }\else{QQ }\fi plot \label{fig:t1e}}
\end{subfigure}%
~
\begin{subfigure}{0.33\textwidth}
    \centering
    \includegraphics[width=\textwidth]{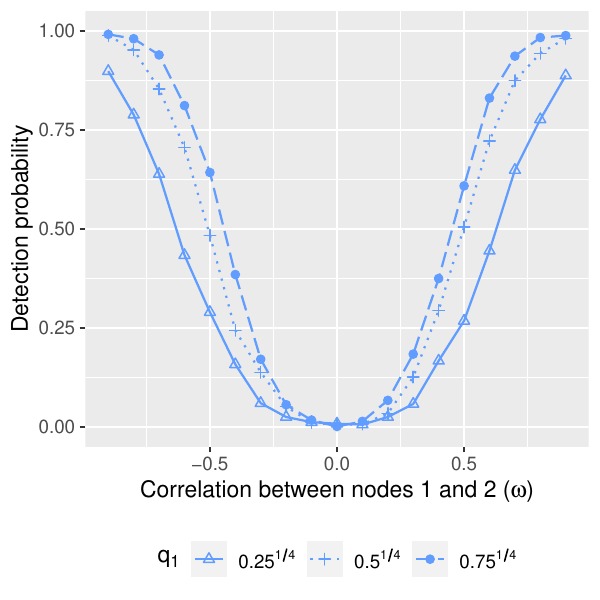}
    \caption{Detection curve (method (c)) \label{fig:detect}}
\end{subfigure}%
~
\begin{subfigure}{0.33\textwidth}
    \centering
    \includegraphics[width=\textwidth]{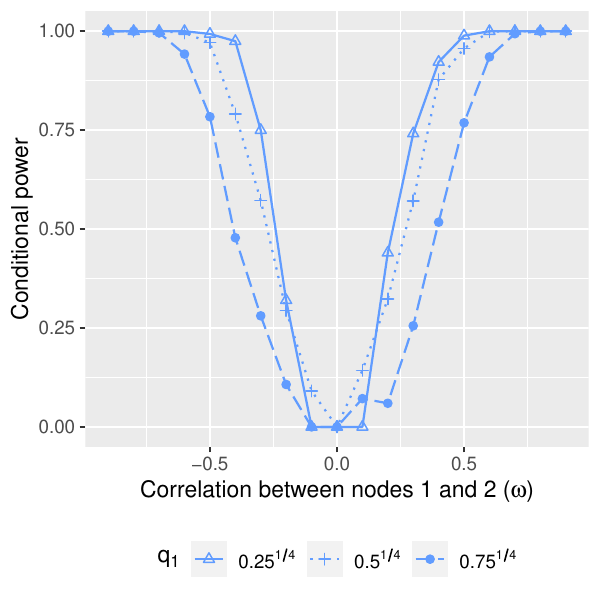}
    \caption{Conditional power (method (c)) \label{fig:cpow}}
\end{subfigure}
\caption{Results for the simulation described in Section \ref{subsec:ex-test}. The left plot demonstrates that only method (c), the decomposition strategy that accounts for dependency between folds, controls the Type 1 error rate. The second and third panels show that as the amount of information allocated to $\Xt{1}$ increases, the detection probability of method (c) increases, and its conditional power decreases. Each curve represents a distinct value of $q_1$, the top-left element of $Q$ (since $Q$ is $2\times2$, the remaining entries follow up to signs; see Supplement \ref{app:sims-test} for the exact values of $Q$).
\label{fig:inf}}
\end{figure}

\subsection{Data analysis}
\label{subsec:ex-eeg}

The electroencephalography dataset from the UCI machine learning repository \citep{misc_eeg_database_121} was originally collected to examine electroencephalography correlates of the genetics of alcoholism. The data collection details are described in detail by \cite{zhang1997electrophysiological}. The full dataset included 122 subjects who belong to either the alcoholic group or the control group. The data for a single subject from a single trial is a matrix with $64$ rows (one per electrode) and $256$ columns (the brain is recorded for one second at 256 Hz). After standardizing the rows to have mean $0$ and variance $1$, we model this data as $X \sim N_{64\times256}(\mathbf{0}_{64\times256}, \Delta, \Gamma(\rho))$ where $\Delta$ is a correlation matrix and $\Gamma(\rho)$ is a first-order autoregressive covariance matrix. Our goal is to identify and evaluate clusters of connectivity in one subject's brain, defined as blocks of non-zero entries in $\Delta$, using a single trial. Because $n=1$, sample splitting is not an option. 

Following the outline in Example \ref{ex:validate}, we study $\Delta$ with the following three-step process:
\begin{enumerate}
    \item Apply Algorithm \ref{alg:premult} with $r=1$, $K=2$, $Q=\begin{pmatrix} 0.5^{1/4} & \sqrt{1-0.5^{1/2}} \\ \sqrt{1-0.5^{1/2}} & -0.5^{1/4} \end{pmatrix}$, and $\Sigma'=I_{64\times256}$ to decompose $\text{vec}(X)$ into $\Xt{1}$ and $\Xt{2}$. 
    \item Estimate $\hat\Delta$ and $\hat\rho$ from $\Xt{1}$; details are given in Algorithm \ref{alg:eeg} in Supplement \ref{app:sims-eeg}. Then, perform hierarchical clustering on the electrodes, where distance between any pair of electrodes is defined as one minus the value of the corresponding entry of $\hat\Delta$.  Let $\mathcal{C}(h)$ be the clustering from step $h$, where $h=1,\dots,64$.
    \item Identify the optimal clustering $\mathcal{C}(\hat h)$ as $\hat h=\arg\max_h \mathcal{L}(\hat\Delta(h),\hat\rho;\Xt{2}|\Xt{1})$ where $\hat\Delta(h)$ sets to zero all entries of $\hat\Delta$ that correspond to electrodes \emph{not} in the same cluster in $\mathcal{C}(h)$. 
\end{enumerate}

Following this procedure for control subject co2c0000337, 
we identify five clusters of electrodes; see Figure \ref{fig:eeg}. Panel \ref{fig:cllX2X1} displays the conditional log-likelihood $\mathcal{L}(\hat\Delta(h),\hat\rho;\Xt{2}|\Xt{1})$ as a function of the number of clusters $h$; it is maximized at $\hat h=5$. 
Panel \ref{fig:brain} displays the electrode labels in accordance with their placement on the subject's head, and colours the labels by cluster. The three large clusters represent three distinct regions of the brain: the left, front-right, and back. The two smaller clusters suggest a pocket in the centre-right with limited external communication. 

\begin{figure}
\centering
\begin{subfigure}{0.35\textwidth}
    \centering
    \includegraphics[width=\textwidth]{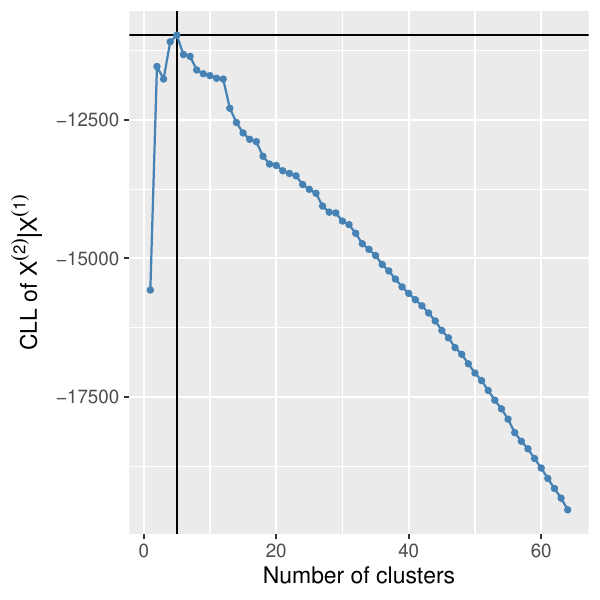}
    \caption{Validation fold CLL \label{fig:cllX2X1}}
\end{subfigure}%
~
\begin{subfigure}{0.35\textwidth}
    \centering
    \includegraphics[width=\textwidth]{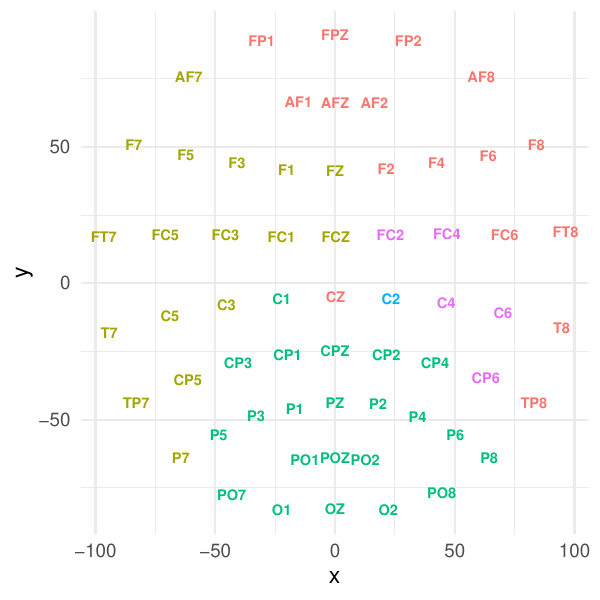}
    \caption{Spatial organization of clusters \label{fig:brain}}
\end{subfigure}%
\caption{Results for the data analysis in Section \ref{subsec:ex-eeg}. Panel (a): The conditional log-likelihood (CLL) curve of $\hat\Delta$ and $\hat\rho$ with respect to  $\Xt{2}|\Xt{1}$ as a function of the number of clusters. The vertical and horizontal lines indicate the optimal CLL value and the corresponding number of clusters. Panel (b): Map of electroencephalography electrode locations (electrodes ``X", ``Y", and ``nd" are excluded as they refer to the two electrooculography and reference electrodes, respectively). Positive x- and y-values refer to the right and front of the subject's head, respectively.  
\label{fig:eeg}}
\end{figure}

\section{Discussion}

Randomization strategies are an increasingly useful alternative to sample splitting in settings where the latter is either inadequate or inapplicable. These strategies fall into two broad classes: those that lead to independent folds and those that lead to dependent folds. When available, independent decompositions are preferable, as they typically can directly replace sample splitting in data analysis pipelines with only minor modifications. By contrast, dependent decompositions trade ease of application in exchange for an astounding degree of flexibility. In the case of multivariate Gaussians, we have shown that \emph{only} a dependent randomization strategy is available in the $n=1$ setting. The price of dependence is a complicated conditional likelihood with a different covariance from the original data. 

While in this paper we focused on multivariate Gaussians, one could apply Algorithm \ref{alg:premult} to any data: once we relax the goal of independence, it only remains to identify the distributions of $\Xt{1}$ and $\Xt{2}|\Xt{1}$. As a concrete example, if $X$ is multivariate Cauchy, then $\Xt{1}$ will follow a multivariate version of the Voigt distribution, and while $\Xt{2}|\Xt{1}$ will be some complex unnamed conditional distribution, it will be tractable in the sense that we can evaluate its likelihood. Identifying the set of distributions for which Algorithm \ref{alg:premult} may prove useful is a potentially interesting research question. Finally, we comment that following the results of Section \ref{sec:GP}, we expect that randomization strategies for other stochastic processes may be within reach and are worth exploring. 


\section*{Acknowledgement}
We thank Jordan Bryan for helpful suggestions that contributed to a major re-framing of an earlier version of this work. We thank Daniel Kessler for helpful conversations about the electroencephalograpy data analysis.
We acknowledge funding from the following sources: 
Office of Naval Research of the United States, Simons Foundation,  Keck Foundation,  National Science Foundation, and  National Institutes of Health of the United States to DW; National Institutes of Health of the United States to DW and JB; and Natural Sciences and Engineering Research Council of Canada to LG and AD.

\section*{Supplementary material}
\label{SM}
The Supplementary Material includes details on decomposing into dependent folds when $n>1$, proofs of technical results, hyperparameter selection, and additional details on the simulation and data analyses.

\bibliographystyle{agsm}
\bibliography{gaussians}

\newpage
\setcounter{figure}{0}
\setcounter{lemma}{0}
\setcounter{theorem}{0}
\setcounter{algo}{0}
\setcounter{proposition}{0}
\setcounter{remark}{0}
\setcounter{example}{0}
\setcounter{corollary}{0}
\setcounter{definition}{0}
\makeatletter
\renewcommand \thefigure{S\@arabic\c@figure}
\renewcommand \thelemma{S\@arabic\c@lemma}
\renewcommand \thetheorem{S\@arabic\c@theorem}
\renewcommand \thealgo{S\@arabic\c@algo}
\renewcommand \theproposition{S\@arabic\c@proposition}
\renewcommand \theremark{S\@arabic\c@remark}
\renewcommand \theexample{S\@arabic\c@example}
\renewcommand \thecorollary{S\@arabic\c@corollary}
\renewcommand \thedefinition{S\@arabic\c@definition}
\makeatother
\appendix
\section*{Supplementary Materials}
\section{Theorem~\ref{thm:dependent} when $n>1$} 
\label{app:thm:dependent:nbig}

When applying Algorithm \ref{alg:premult}, we would like the structure of $\Xt{1},\dots,\Xt{K}$ to resemble that of $X$, so  that models designed for $X$ can be easily extended to $\Xt{1},\dots,\Xt{K}$. When $n>1$, the rows of $X$ are independent and identically distributed, but as seen in Theorem \ref{thm:dependent}, the rows of $X'$ resulting from Algorithm~\ref{alg:premult} are not. The following corollary of Theorem~\ref{thm:dependent} shows that if one chooses $r$ to be a multiple of $n$ and $Q$ to be a block diagonal matrix, then one can construct $\Xt{1},\dots,\Xt{K}$ so that the rows \emph{within} each fold are independent, though there is dependence \emph{between} folds. 

\begin{corollary}
\label{cor:block}
Consider the setting of Theorem \ref{thm:dependent} where $n>1$, $r=r'n$ where $r'\in \mathbb{Z}_+$, $K=r'+1$, $Q=I_n \otimes Q'$ where $Q'$ is an orthogonal matrix of dimension $K\times K$, and $X^\aug$ is constructed as in Remark \ref{remark:Xaug}. Let $Q_X$ denote the first column of $Q$. Then, marginally,
$$
\text{vec}(X'^\top) \sim N_{nKp}\left(\text{vec}(\mu1_n^\top\otimes Q_X^\top),I_n\otimes (Q_XQ_X^\top\otimes\Sigma+(I_K - Q_XQ_X^\top)\otimes\Sigma')\right).
$$
For $k=1,\dots,K$, let $\Xt{k}$ be the $n \times p$ submatrix of $X'$ consisting of the $(k+Ki)$th rows where $i\in\{0,\dots,n-1\}$. Then,
$$
\Xt{k}\sim N_{n\times p}(q_k'1_n\mu^\top,I_n,q_k'^2\Sigma+(1-q_k'^2)\Sigma'),
$$
where $q_k'$ is the $k$th entry of the first column of $Q'$. In other words, the $n$ rows of $\Xt{k}$ are independent and identically distributed multivariate Gaussians with mean $q_k'\mu$ and covariance $q_k'^2\Sigma+(1-q_k'^2)\Sigma'$.
\end{corollary}

The procedure described in Corollary \ref{cor:block} can be viewed as applying Theorem \ref{thm:dependent} to each row of $X$ separately, and then taking the $k$th row from each of the $n$ output matrices of Algorithm \ref{alg:premult}, and grouping them together for form $\Xt{k}$, for $k=1,\dots,K$. Thus, each $\Xt{k}$ matrix contains exactly one row originating from each of the rows of $X$, thereby replicating the structure of $X$.

\section{Proofs}\label{sec:proofs}

\subsection{Proof of Proposition \ref{prop:thin}}

\begin{proof}
Since the auxiliary noise vectors have the same covariance as $X$, the distribution of $X^\aug$ can be expressed concisely as
$$
X^\aug \sim N_{K\times p}\left(\begin{pmatrix} \mu^\top \\ 0_{(K-1) \times p} \end{pmatrix} ,I_K,\Sigma\right).
$$

Next, let $\epsilon=(\sqrt{\epsilon_1},\dots,\sqrt{\epsilon_K})^\top$, which implies that $Q$ can be written as $Q=\begin{pmatrix}\epsilon, U^\perp_\epsilon\end{pmatrix}$. Observe that $\|\epsilon\|^2=\sum_{k=1}^K\epsilon_k=1$, $Q$ is an orthogonal matrix and  
$$
X'=QX^\aug\sim N_{K\times p}\left(\begin{pmatrix}\epsilon & U^\perp_\epsilon\end{pmatrix}\begin{pmatrix} \mu^\top \\ 0_{(K-1) \times p} \end{pmatrix},I_K,\Sigma\right) = N_{K\times p}\left(\epsilon\mu^\top ,I_K,\Sigma\right).
$$

From the distribution of $X'$, we can recover the results in the proposition. 
Since the row-covariance is $I_K$, it immediately follows that the rows of $X'$ are independent. The distribution of the $k$th row is then given by $\Xt{k}\sim N_p(\sqrt{\epsilon_k}\mu,\Sigma)$ as required.

\end{proof}

\subsection{Proof of Proposition \ref{prop:fission}}

\begin{proof}
Write $W\sim N_p(0,\Sigma')$ to denote the single auxiliary noise vector. Then,
$$
X'=QX^\aug=\begin{pmatrix} \frac{1}{\sqrt{2}} & \frac{1}{\sqrt{2}} \\ \frac{1}{\sqrt{2}} & -\frac{1}{\sqrt{2}} \end{pmatrix}\begin{pmatrix}X^\top \\ W^\top\end{pmatrix} = \begin{pmatrix} \frac{1}{\sqrt{2}}(X^\top + W^\top) \\ \frac{1}{\sqrt{2}} (X^\top-W^\top)\end{pmatrix}.
$$
Since $X'$ is jointly normal, it remains to find the mean and variance for each row, and the covariance between the rows. We have that 
\begin{align*}
    \E\left[\frac{1}{\sqrt{2}}(X^\top + W^\top)\right] &= \frac{1}{\sqrt{2}}\mu, \\
    \E\left[\frac{1}{\sqrt{2}}(X^\top - W^\top)\right] &= \frac{1}{\sqrt{2}}\mu, \\
    \Var\left(\frac{1}{\sqrt{2}}(X^\top + W^\top)\right) &= \frac{1}{2}(\Sigma + \Sigma'), \\
    \Var\left(\frac{1}{\sqrt{2}}(X^\top - W^\top)\right) &= \frac{1}{2}(\Sigma + \Sigma'), \\
    \Cov\left(\frac{1}{\sqrt{2}}(X^\top + W^\top),\frac{1}{\sqrt{2}}(X^\top - W^\top)\right) &= \frac{1}{2}(\Sigma - \Sigma').
\end{align*}
Thus,
$$
\begin{pmatrix} \Xt{1} \\ \Xt{2} \end{pmatrix} \sim N_{2p}\left(\frac{1}{\sqrt{2}}\begin{pmatrix} \mu \\ \mu \end{pmatrix}, \frac{1}{2}\begin{pmatrix} \Sigma + \Sigma' & \Sigma - \Sigma' \\ \Sigma - \Sigma' & \Sigma + \Sigma' \end{pmatrix} \right)
$$
as required.
\end{proof}

\subsection{Proof of Proposition \ref{prop:premult}}

\begin{proof}
Since $r=0$, $X^{aug}=X$. The first conclusion follows immediately from the basic properties of the matrix-variate Gaussian and the fact that $Q$ is orthogonal. Specifically, $X'=QX\sim N_{n\times p}(Q1_n\mu^\top,QQ^\top,\Sigma)$ and $QQ^\top=I_n$. For the second conclusion, the constraint that $Q1_n=1_n$ then implies that the distribution of $X'$ reduces to $X'\sim N_{n\times p}(1_n\mu^\top,I_n,\Sigma)$ as required.
\end{proof}

\subsection{Proof of Proposition \ref{prop:premult-fisher}}

\begin{proof}
Since the rows of $X'$ are independent, it suffices to compute the Fisher information allocated to each row of $X'$, and then add the information in the appropriate $n_k$ rows to recover the Fisher information allocated to $\Xt{k}$. 

Let $I_{X_1}(\mu)$ and $I_{X_1}(\Sigma)$ denote the amount of Fisher information contained in a single row of $X$ about $\mu$ and $\Sigma$, respectively. Since the rows of $X$ are independent and identically distributed, the amount of Fisher information about $\mu$ and $\Sigma$ in all of $X$ is then $I_X(\mu)=nI_{X_1}(\mu)$ and $I_X(\Sigma)=nI_{X_1}(\Sigma)$ respectively.

For $j=1,\dots,n$, let $Q_j$ and $X'_j$ denote the $j$th rows of $Q$ and $X'$, respectively. Note that $X'_j\sim N_p(\mu 1_n^\top Q_j,\Sigma)$. It then follows from standard Fisher information manipulations that $I_{X'_j}(\mu)=(1_n^\top Q_j Q_j^\top1_n)I_{X_1}(\mu)$ and $I_{X'_j}(\Sigma)=I_{X_1}(\Sigma)$. Since the rows of $X'$ are independent, aggregating the appropriate $X'_j$ into each $\Xt{k}$ implies that $I_{\Xt{k}}(\mu)=[(1_n^\top Q^{(k)\top} \Qt{k}1_n)/n] I_{X}(\mu)$ and $I_{\Xt{k}}(\Sigma)=[n_k/n] I_{X}(\Sigma)$ as required.

\end{proof}

\subsection{Proof of Theorem \ref{thm:maxn}}

Theorem \ref{thm:maxn} is closely connected to the concept of data thinning introduced in \cite{neufeld2023data} and \cite{dharamshi2023generalized}: its conclusion is equivalent to the claim that $X$ cannot be thinned. To see this, we restate their definition of data thinning. 

\begin{definition}[Data thinning \citep{dharamshi2023generalized}] 
\label{def:thinning}
Consider a family of distributions $\mathcal P=\{P_\theta:\theta\in\Omega\}$. For $Y \sim P_\theta$, suppose that we can sample $(\Yt{1},\ldots,\Yt{K})|Y$, without knowledge of $\theta$, to achieve the following properties:
\begin{enumerate}
  \item $\Yt{1},\ldots,\Yt{K}$ are mutually independent (with distributions depending on $\theta$),  and
    \item $Y=T(\Yt{1},\ldots,\Yt{K})$ for some deterministic function $T(\cdot)$.
\end{enumerate}
Then, we refer to the process of sampling  $(\Yt{1},\ldots,\Yt{K})|Y$ as \emph{thinning} by the function $T(\cdot)$. 
\end{definition}

We begin the proof with a technical lemma about data thinning.

\begin{lemma}
\label{lem:thin}
Suppose that $X\sim P_\theta$, where $P_\theta$ is a distribution parameterized by an unknown $\theta$, and that $S(X)$ is a sufficient statistic for $\theta$. If $S(X)$ cannot be thinned, then neither can $X$.
\end{lemma}

\begin{proof}
We will establish the result by proving the contrapositive. 

Suppose that $X\sim P_\theta$ can be thinned. Then, there must exist two distributions parameterized by $\theta$, $\Qt{1}_\theta$ and $\Qt{2}_\theta$, and a deterministic function $T(\cdot)$ such that the following hold:
\begin{enumerate}
    \item $(\Xt{1},\Xt{2})\sim \Qt{1}_\theta \times \Qt{2}_\theta$,
    \item $X = T(\Xt{1},\Xt{2})$, and
    \item $G_t$, the conditional distribution of $(\Xt{1},\Xt{2})$ given $T(\Xt{1},\Xt{2})=t$ does not depend on $\theta$.
\end{enumerate}

Substituting in $X=T(\Xt{1},\Xt{2})$, we see that $\tilde T(\Xt{1},\Xt{2}):=S(T(\Xt{1},\Xt{2}))=S(X)$. It thus remains to show that $\tilde G_{\tilde t}$, the conditional distribution of $(\Xt{1},\Xt{2})$ given $\tilde T(\Xt{1},\Xt{2})=\tilde t$, does not depend on $\theta$.

Consider the distribution of $(\Xt{1},\Xt{2})|\{\tilde T(\Xt{1},\Xt{2})=\tilde t\}$. It follows from the law of total probability that it can be expressed in terms of the distributions of
\begin{enumerate}
    \item $(\Xt{1},\Xt{2})|\{T(\Xt{1},\Xt{2})=t, \tilde T(\Xt{1},\Xt{2})=\tilde t\}$, for all $t$ such that $\tilde t = S(t)$, and
    \item $\{ T(\Xt{1},\Xt{2})\}|\{ \tilde T(\Xt{1},\Xt{2})=\tilde t\}$.
\end{enumerate} 

For the first term, note that the former conditioning event implies the latter. Dropping this redundant conditioning event yields $G_t$, which does not depend on $\theta$. For the second term, we can rewrite it as the distribution of $X|\{S(X)=\tilde t\}$, which does not depend on $\theta$ since $S(X)$ is sufficient for $\theta$. Thus, $\tilde G_{\tilde t}$ does not depend on $\theta$, thereby proving the result.
\end{proof}

The lemma implies that instead of proving that $X$ cannot be thinned, we can instead focus on showing that a sufficient statistic for its parameters cannot be thinned. We focus on the case where $\mu$ is known, as proving that one cannot thin when $\mu$ is known implies that one cannot thin when $\mu$ is unknown. When $\mu$ is known, $S(X)=(X-\mu)(X-\mu)^\top\sim\text{Wishart}_p(1,\Sigma)$  (i.e. $S(X)$ is a singular Wishart random matrix with one degree of freedom; see \cite{srivastava2003singular} for technical details on the singular Wishart distribution) is sufficient for $\Sigma$. 

The Wishart family is a natural exponential family, so Theorem 3 of \cite{dharamshi2023generalized} says that if it can be thinned then the thinning function must be addition. The next lemma shows that in the special case of one degree of freedom, the Wishart distribution cannot be thinned.

\begin{lemma}[A $\text{Wishart}_p(1, \Sigma)$ cannot be thinned] \label{lem:nondecomposable}
A singular Wishart with $n=1$ cannot be thinned in the sense of Definition~\ref{def:thinning}.
\end{lemma}

\begin{proof}
Consider the random matrix $W\sim\text{Wishart}_p(1,\Sigma)$. Classical results regarding the Wishart distribution state that $W$ is indecomposable \citep{lévy_1948, SHANBHAG1976347, Peddada1991}. Since there do not exist two non-trivial random matrices that sum to $W$, there cannot exist an additive thinning function. The result then follows from the contrapositive of Theorem 3 of \cite{dharamshi2023generalized}.
\end{proof}

We will now prove the theorem by contradiction. 

\begin{proof}
Suppose that $X$ can be non-trivially decomposed into multiple independent random variables, that is, it can be thinned. Then the contrapositive of Lemma \ref{lem:thin} implies that $S(X)$ can also be thinned. However, since $S(X)\sim\text{Wishart}_p(1,\Sigma)$, Lemma \ref{lem:nondecomposable} proves that $S(X)$ cannot be thinned, a contradiction. Therefore, $X$ cannot be thinned.
\end{proof}

\subsection{Proof of Theorem \ref{thm:dependent}}

\begin{proof}
Let $W\sim N_{r\times p}(0,I_r,\Sigma')$ be the matrix of noise vectors generated to construct $X^\aug$, and let $Q_W$ denote all columns of $Q$ excluding the first. Note that $Q_XQ_X^\top + Q_WQ_W^\top = I_K$. Then, we can write $X'$ and $\text{vec}(X'^\top)$ as
\begin{align*}
    X' &= QX^\aug = Q_XX^\top + Q_WW \\
    \implies \text{vec}(X'^\top) &= \text{vec}(X Q_X^\top) + \text{vec}(W^\top Q_W^\top).
\end{align*}

Since $\text{vec}(X Q_X^\top)$ and $\text{vec}(W^\top Q_W^\top)$ are by construction independent Gaussian vectors of length $Kp$, it remains to identify and add their parameters. 
Starting with the mean, by linearity of expectation, $E[\text{vec}(X'^\top)]=\text{vec}(E[X'^\top])=\text{vec}(E[X]Q_X^\top)=\text{vec}(\mu Q_X^\top)$, and then the covariance,
\begin{align*}
\text{Var}(\text{vec}(X'^\top)) &= \text{Var}(\text{vec}(X Q_X^\top)) + \text{Var}(\text{vec}(W^\top Q_W^\top)) \\
&= \text{Var}((Q_X\otimes I_p)\text{vec}(X)) + \text{Var}((Q_W\otimes I_p)\text{vec}(W^\top)) \\
&= (Q_X\otimes I_p)\text{Var}(X)(Q_X^\top\otimes I_p) + (Q_W\otimes I_p)\text{Var}(\text{vec}(W^\top))(Q_W^\top\otimes I_p) \\
&= (Q_X\otimes I_p)(1\otimes\Sigma)(Q_X^\top\otimes I_p) + (Q_W\otimes I_p)(I_r\otimes\Sigma')(Q_W^\top\otimes I_p) \\
&= Q_XQ_X^\top \otimes \Sigma + Q_WQ_W^\top\otimes\Sigma' \\
&= Q_XQ_X^\top \otimes \Sigma + (I_K-Q_XQ_X^\top)\otimes\Sigma'.
\end{align*}
Thus, $\text{vec}(X'^\top)\sim  N_{Kp}\left(\text{vec}(\mu Q_X^\top),  Q_XQ_X^\top \otimes \Sigma + (I_K-Q_XQ_X^\top) \otimes \Sigma'\right)$ as required.
\end{proof}

\subsection{Proof of Theorem \ref{thm:fisher}}

\begin{proof}
From Remark \ref{remark:conditional}, we have that $\Xt{1}\sim N_p(q_1\mu(\theta),q_1^2\Sigma(\phi) + (1-q_1^2)\Sigma')$. Standard results about the Fisher information of a multivariate Gaussian \citep{mardia1984maximum} yield that 
\begin{align*}
\left[I_{\Xt{1}}(\theta)\right]_{ii'} &= \left(\frac{\partial q_1\mu(\theta)^\top}{\partial \theta_i}\right)\left(q_1^2\Sigma(\phi) + (1-q_1^2)\Sigma'\right)^{-1}\left(\frac{\partial q_1\mu(\theta)}{\partial \theta_{i'}}\right) \\
&= q_1^2\left(\frac{\partial \mu(\theta)^\top}{\partial \theta_i}\right)\left(q_1^2\Sigma(\phi) + (1-q_1^2)\Sigma'\right)^{-1}\left(\frac{\partial \mu(\theta)}{\partial \theta_{i'}}\right), \\
\\
\left[I_{\Xt{1}}(\phi)\right]_{jj'} &= \frac{1}{2}\trace\left[\left(q_1^2\Sigma(\phi) + (1-q_1^2)\Sigma'\right)^{-1}\left(\frac{\partial \left(q_1^2\Sigma(\phi) + (1-q_1^2)\Sigma'\right)}{\partial \phi_j}\right)\right. \\
&\quad\quad\quad\quad\left.\left(q_1^2\Sigma(\phi) + (1-q_1^2)\Sigma'\right)^{-1}\left(\frac{\partial \left(q_1^2\Sigma(\phi) + (1-q_1^2)\Sigma'\right)}{\partial \phi_{j'}}\right)\right] \\
&= \frac{1}{2}\trace\left[\left(q_1^2\Sigma(\phi) + (1-q_1^2)\Sigma'\right)^{-1}\frac{\partial q_1^2\Sigma(\phi)}{\partial \phi_j}\left(q_1^2\Sigma(\phi) + (1-q_1^2)\Sigma'\right)^{-1}\frac{\partial q_1^2\Sigma(\phi)}{\partial \phi_{j'}}\right] \\
&= \frac{1}{2}q_1^4\trace\left[\left(q_1^2\Sigma(\phi) + (1-q_1^2)\Sigma'\right)^{-1}\frac{\partial \Sigma(\phi)}{\partial \phi_j}\left(q_1^2\Sigma(\phi) + (1-q_1^2)\Sigma'\right)^{-1}\frac{\partial \Sigma(\phi)}{\partial \phi_{j'}}\right].
\end{align*}

The forms of $\left[I_{\Xt{2}|\Xt{1}}(\theta)\right]_{ii'}$ and $\left[I_{\Xt{2}|\Xt{1}}(\phi)\right]_{jj'}$ then follows from the fact that $\left[I_X(\theta)\right]_{ii'}=\left[I_{\Xt{1}}(\theta)\right]_{ii'}+\left[I_{\Xt{2}|\Xt{1}}(\theta)\right]_{ii'}$ and $\left[I_X(\phi)\right]_{jj'}=\left[I_{\Xt{1}}(\phi)\right]_{jj'}+\left[I_{\Xt{2}|\Xt{1}}(\phi)\right]_{jj'}$ since Fisher information is preserved by Algorithm \ref{alg:premult}.
\end{proof}

\subsection{Proof of Proposition \ref{prop:GPthin}}

\begin{proof}
Writing $\Xt{k}(t) = Q_{k1}X(t) + \sum_{j=1}^r Q_{k(1+j)}W_j(t)$, that is, as a linear combination of independent Gaussian processes, it follows from the basic properties of Gaussian processes that $\Xt{k}(t)\sim GP(Q_{k1}\mu,(\sum_{j=1}^{r+1}Q_{kj}^2)C) = GP(\sqrt{\epsilon_k}\mu,C)$.

It remains to show that $\Xt{1}(t),\dots,\Xt{K}(t)$ are independent. To see this, consider any finite index set $T_d=\{t_1,\dots,t_d\}\in T$ and construct the random vectors $(\Xt{k}(t_1),\dots,\Xt{k}(t_d))=(Q_{k1}X(t_1)+ \sum_{j=1}^r Q_{k(1+j)}W_j(t_1),\dots,Q_{k1}X(t_d)+ \sum_{j=1}^r Q_{k(1+j)}W_j(t_d))$ for all $k\in\{1,\dots,K\}$. We will write this as $\Xt{k}(T_d)=Q_{k1}X(T_d) + \sum_{j=1}^r Q_{k(1+j)}W_j(T_d)$ for convenience. Let $\Sigma_{T_d}$ be the covariance matrix constructed by evaluating $C(t,t')$ at every pair of points in $T_d$. Then, the covariance between any pair of vectors $\Xt{k}(T_d)$ and $\Xt{k'}(T_d)$ for $1\le k < k' \le K$ is:
\begin{align*}
\text{Cov}\left(\Xt{k}(T_d),\Xt{k'}(T_d)\right) &=\text{Cov}\left(Q_{k1}X(T_d) + \sum_{j=1}^r Q_{k(1+j)}W_j(T_d),Q_{k'1}X(T_d) + \sum_{j=1}^r Q_{k'(1+j)}W_j(T_d)\right) \\
&= Q_{k1}Q_{k'1}\text{Var}\left(X(T_d)\right) + \sum_{j=1}^r Q_{k(1+j)}Q_{k'(1+j)}\text{Var}\left(W_j(T_d)\right) \\
&= \sum_{j=1}^{r+1} Q_{kj}Q_{k'j}\Sigma_{T_d} \\
&= 0.
\end{align*}
Since the covariance is $0$ between all pairs of folds for all finite index sets, it follows that $\Xt{1}(t),\dots,\Xt{K}(t)$ are independent as required.
\end{proof}

\subsection{Proof of Theorem \ref{thm:dependent-GP}}

\begin{proof}
Writing $\Xt{k}(t) = Q_{k1}X(t) + \sum_{j=1}^r Q_{k(1+j)}W_j(t)$, that is, as a linear combination of independent Gaussian processes, it follows from the basic properties of Gaussian processes that marginally, $\Xt{k}(t)\sim GP(Q_{k1}\mu,Q_{k1}^2C+(\sum_{j=1}^{r}Q_{k(1+j)}^2)C') = GP(q_k\mu,q_k^2C + (1-q_k^2)C')$.

To show the conditional result, first consider any finite index set $T_d=\{t_1,\dots,t_d\}\in T$ and construct the random vectors $\Xt{1}(T_d)=(\Xt{1}(t_1),\dots,\Xt{1}(t_d))$ and $\Xt{2}(T_d)=(\Xt{2}(t_1),\dots,\Xt{2}(t_d))$. These random vectors are jointly Gaussian as they are linear combinations of the same set of underlying independent Gaussians:
$$
\begin{pmatrix} \Xt{1}(T_d) \\ \Xt{2}(T_d) \end{pmatrix} \sim N_{2d}\left(\begin{pmatrix}q_1\mu_{T_d} \\ q_2\mu_{T_d}\end{pmatrix}, 
\begin{pmatrix}q_1^2\Sigma_{T_d}+(1-q_1^2)\Sigma'_{T_d} & q_1q_2(\Sigma_{T_d}-\Sigma'_{T_d}) \\ q_1q_2(\Sigma_{T_d}-\Sigma'_{T_d}) & q_2^2\Sigma_{T_d}+(1-q_2^2)\Sigma'_{T_d} \end{pmatrix}\right)
$$
The result then follows from standard results about the conditional distribution of joint Gaussians.
\end{proof}
\section{Hyperparameter selection when $n=1$}
\label{app:tuning}

In Section \ref{subsec:fission-computation}, we discuss the computational benefits of setting $\Sigma'=\sigma'^2I_p$ where $\sigma' > 0$ is a constant chosen by the analyst. Here we discuss considerations for choosing the values of $Q$ and $\sigma'$. Our overarching goal is to achieve predictable allocations of Fisher information about $\Sigma$ across folds.

We focus on the case when $K=2$. Suppose that we apply Algorithm \ref{alg:premult} to a single sample of $X\sim N_p(\mu,\Sigma)$ with $r=1$, $K=2$, $\Sigma'=\sigma'^2I_p$, and $Q$ a $2\times 2$ orthogonal matrix. Our goal is to allocate $\gamma\in(0,1)$ of the total Fisher information about each of the $N$ parameters in $\Sigma$ to $\Xt{1}$ (i.e. $\gamma[I_{X}(\phi)]_{jj'} = [I_{\Xt{1}}(\phi)]_{jj'}$ for all $1 \le j,j' \le N$). 

Using the expressions given in Theorem \ref{thm:fisher}, we can rewrite our goal as
$$
\gamma\trace\left[\Sigma(\phi)^{-1}\frac{\partial \Sigma(\phi)}{\partial \phi_j} \Sigma(\phi)^{-1}\frac{\partial \Sigma(\phi)}{\partial \phi_{j'}} \right] = q_1^4\trace\left[\left(q_1^2\Sigma(\phi)+(1-q_1^2)\sigma'^2I_p\right)^{-1} \frac{\partial \Sigma(\phi)}{\partial \phi_j} \left(q_1^2\Sigma(\phi)+(1-q_1^2)\sigma'^2I_p\right)^{-1} \frac{\partial \Sigma(\phi)}{\partial \phi_{j'}} \right]
$$
for all $1 \le j,j' \le N$.

In order to achieve the target Fisher information allocation, we must identify a pair $(q_1,\sigma')$ that solves the above system of equations. As this is a system of $K(K+1)/2$ equations in two unknowns, an exact solution will not in general be available unless $\Sigma(\phi)$ is diagonal. Furthermore, as established in Section \ref{subsec:fission-fisher}, the Fisher information expression depends on the unknown $\Sigma$, and so, without knowing $\Sigma$, we will not be able to choose a $Q$ and $\sigma'$ that satisfy an exact allocation. However, certain choices of $Q$ and $\sigma'$ may be preferable. 

We first suggest to simplify the problem by fixing $q_1=\gamma^{1/4}$. Then, we propose to choose $\sigma'$ by solving the following univariate minimization problem:
\begin{align*}
\sigma'=&\arg\min_{c>0} \sum_{1\le j<j' \le N}
\left(\trace\left[\Sigma(\phi)^{-1}\frac{\partial \Sigma(\phi)}{\partial \phi_j} \Sigma(\phi)^{-1}\frac{\partial \Sigma(\phi)}{\partial \phi_{j'}} - \right.\right. \\
&\left.\left.\left(\gamma^{1/2}\Sigma(\phi)+(1-\gamma^{1/2})c^2I_p\right)^{-1} \frac{\partial \Sigma(\phi)}{\partial \phi_j} \left(\gamma^{1/2}\Sigma(\phi)+(1-\gamma^{1/2})c^2I_p\right)^{-1} \frac{\partial \Sigma(\phi)}{\partial \phi_{j'}} \right]\right)^2.
\end{align*}

As the objective function depends on $\Sigma(\phi)$, it cannot be solved exactly. Instead, we suggest to replace $\Sigma(\phi)$ with an a priori guess, $S$, and then use standard numerical optimization software to estimate $\sigma'$. In Section \ref{sec:sims}, since the variances are known to all equal $1$, for convenience we choose $S=I_p$ leading to $\sigma'=1$.

We conclude by commenting that when $K>2$, in the interest of preserving symmetry for tasks such as cross-validation, we often will choose $q_1=\dots=q_K=1/\sqrt{K}$. One can then choose $\sigma'$ as in the $K=2$ case.

\section{Additional simulation and data analysis details}\label{app:sims}

\subsection{Details of Section \ref{subsec:ex-test}}
\label{app:sims-test}

Algorithm \ref{alg:inference} summarises the three methods used to conduct inference on the thresholded entries of $\Delta$. We combine all three methods into one algorithm and use substeps to indicate differences by approach.

\begin{algo}[Inference after selecting the largest entry of a matrix-variate Gaussian row-covariance]
\label{alg:inference}
Start with a single realization of $X\sim N_{a\times b}(\mathbf{0}_{a,b},\Delta,\Gamma(\rho))$ and a constant $q_1\in(0,1)$. Then, perform the following:
\begin{enumerate}
    \item Apply Algorithm \ref{alg:premult} with $r=1$, $K=2$, $Q=\begin{pmatrix}q_1^2 & \sqrt{1-q_1^2} \\ \sqrt{1-q_1^2} & -q_1^2\end{pmatrix}$, and $\Sigma'=I_{ab}$ to decompose $\text{vec}(X)\sim N_{ab}(\mathbf{0}_{ab},\Gamma(\rho)\otimes\Delta)$ into $\Xt{1}$ and $\Xt{2}$. See Remark \ref{remark:conditional} for details on the marginal and conditional distributions.
    
    \item Compute a starting estimate of $\hat\Delta$: 
    \begin{enumerate}[]
        \item Method (a): Set $\hat\Delta$ to the sample row-covariance of $X$.
        \item Methods (b) and (c): Construct $\text{mat}(\Xt{1}; a, b)$, where $\text{mat}(\cdot; a,b)$ is the function that reshapes a length $ab$ vector into an $a\times b$ matrix in column order. Set $\hat\Delta$ to the sample row-covariance of $\text{mat}(\Xt{1})$.
    \end{enumerate}
    \item Select the entry of $\hat\Delta$ with the largest magnitude, and denote its counterpart in $\Delta$ as $\delta$. 
    \item Use a likelihood ratio test to test the hypothesis $H_0:\delta=0$, where the test statistic is given by:
    \begin{enumerate}[{Method} (a):]
        \item $T_{naive} = -2\ln\frac{\sup_{\rho\in(0,1),\Delta\in\Theta_0}\mathcal{L}(\Delta,\rho;X)}{\sup_{\rho\in(0,1),\Delta\in\Theta}\mathcal{L}(\Delta,\rho;X)}$ where the log-likelihoods are with respect to the distribution of $X$,
        \item  $T_{marg} = -2\ln\frac{\sup_{\rho\in(0,1),\Delta\in\Theta_0}\mathcal{L}(\Delta,\rho;\Xt{2})}{\sup_{\rho\in(0,1),\Delta\in\Theta}\mathcal{L}(\Delta,\rho;\Xt{2})}$ where the log-likelihoods are with respect to the marginal distribution of $\Xt{2}$,
        \item $T_{cond} = -2\ln\frac{\sup_{\rho\in(0,1),\Delta\in\Theta_0}\mathcal{L}(\Delta,\rho;\Xt{2}|\Xt{1})}{\sup_{\rho\in(0,1),\Delta\in\Theta}\mathcal{L}(\Delta,\rho;\Xt{2}|\Xt{1})}$ where the log-likelihoods are with respect to the conditional distribution of $\Xt{2}|\Xt{1}$, 
    \end{enumerate}
    where $\Theta$ denotes the set of symmetric positive definite matrices with unit diagonal entries, and $\Theta_0$ is the subset of $\Theta$ such that $\delta=0$.
    \item Compute a p-value by comparing $T$ to the quantiles of a $\chi^2_1$ distribution.
\end{enumerate}
\end{algo}

In Step 4 of Algorithm \ref{alg:inference}, we approximate the maximum likelihood on the constrained set $\Theta_0$ by optimizing the log-likelihood over the unconstrained set $\Theta$ subject to a L2 penalty on $\delta$ with penalty parameter 500,000. All likelihoods are optimized using the \texttt{optimize} routine implemented by the 
\texttt{cmdstanr} \texttt{R} interface to the \texttt{Stan} programming language \citep{carpenter2017stan, stan, cmdstanr}.

\subsection{Details of Section \ref{subsec:ex-eeg}}
\label{app:sims-eeg}

Algorithm \ref{alg:eeg} describes the steps taken to estimate $\hat\Delta$ and $\hat\rho$ from $\Xt{1}$ in Section \ref{subsec:ex-eeg}. 

\begin{algo}[Estimating the covariance parameters from noisy electroencephalograpy data]
\label{alg:eeg}
Start with a vector $Y\sim N_{ab}(\mathbf{0}_{ab},q_1^2(\Gamma(\rho)\otimes\Delta)+(1-q_1^2)I_{ab})$ (this is the first fold, $\Xt{1}$, after applying Algorithm \ref{alg:premult} to the vectorised electroencephalograpy data) where $q_1^2\in(0,1)$. Then, perform the following:
\begin{enumerate}
    \item Construct $Y_m=\text{mat}(Y; a, b)$, where $\text{mat}(\cdot; a,b)$ is the function that reshapes a length $ab$ vector into an $a\times b$ matrix in column order. 
    \item To estimate $\hat\Delta$, first estimate the sample row-covariance $S$ from $Y_m$. Then, compute the starting estimate $\Delta^*=\frac{1}{q_1^2}(S-(1-q_1^2)I_{a})$. If $\Delta^*$ is not positive semidefinite, set any negative eigenvalues to $0.1$. Finally, extract the correlation matrix from $\Delta^*$ and denote it by $\hat\Delta$.
    \item To estimate $\hat\rho$, treat the rows of $Y_m$ as independent and compute the maximum likelihood estimate of $\rho$ under a latent AR(1) model with unit innovation variance and $1-q_1^2$ emission variance using the \texttt{dlm} \texttt{R} package \citep{dlm}.
\end{enumerate}
\end{algo}


\subsection{Computation for matrix-variate normal data}
\label{app:fast}

When $X\sim N_{a\times b}(\mathbf{0}_{a,b},\Delta,\Gamma(\rho))$, the eigendecomposition of the covariance of $\mathrm{vec}(X)$ can be computed efficiently. Writing the eigendecompositions of $\Gamma(\rho)$ and $\Delta$ as $PA(\rho)P^\top$ and $RBR^\top$, respectively, we have that $\Gamma(\rho)\otimes\Delta=(P\otimes R)(A(\rho)\otimes B)(P^\top\otimes R^\top)$. Moreover, as $\Gamma(\rho)$ follows a first-order autoregressive model, $P$ is deterministic, and an explicit formula for $A(\rho)$ is available \citep{tridiag}. Thus, \eqref{eq:fastx1}-\eqref{eq:fastx2x1} can be used to efficiently compute likelihoods in this setting.

\end{document}